\documentclass[11pt]{article}
\usepackage{color}
\usepackage[colorlinks=true, allcolors=blue,backref=page]{hyperref}
\usepackage{amsmath, amssymb, amsthm, thmtools}
\usepackage{mathrsfs}
\usepackage{mathtools}
\usepackage[noabbrev,capitalize,nameinlink]{cleveref}
\usepackage{fullpage}
\usepackage[noadjust]{cite}
\usepackage{graphics}
\usepackage{pifont}
\usepackage{tikz}
\usepackage{bbm}
\usepackage{natbib}
\usepackage{multirow}
\usepackage{float} 
\usepackage[T1]{fontenc}

\usetikzlibrary{arrows.meta}

\usepackage{environ}
\usepackage{framed}
\usepackage{url}
\usepackage[linesnumbered,ruled,vlined]{algorithm2e}
\usepackage[noend]{algpseudocode}
\usepackage[labelfont=bf]{caption}
\usepackage{cite}
\usepackage{framed}
\usepackage[framemethod=tikz]{mdframed}
\usepackage{appendix}
\usepackage{graphicx}
\graphicspath{{internal/}}
\usepackage{comment}
\usepackage{booktabs}
\usepackage{tabularx}
\usepackage[textsize=tiny,textwidth=2cm]{todonotes}
\usepackage{tcolorbox}
\usepackage{enumerate}
\allowdisplaybreaks[1]
\usepackage{enumerate}
\usepackage{stmaryrd}
\usepackage{physics}
\usepackage{wasysym}
\usepackage{MnSymbol}
\usepackage[margin=1in]{geometry}
\usepackage{braket}
\usepackage{todonotes}

\usepackage[shortlabels]{enumitem}
\crefformat{enumi}{#2#1#3}
\crefrangeformat{enumi}{#3#1#4 to~#5#2#6}
\crefmultiformat{enumi}{#2#1#3}
{ and~#2#1#3}{, #2#1#3}{ and~#2#1#3}

\DeclareSymbolFont{symbolsC}{U}{pxsyc}{m}{n}
\SetSymbolFont{symbolsC}{bold}{U}{pxsyc}{bx}{n}
\DeclareFontSubstitution{U}{pxsyc}{m}{n}
\DeclareMathSymbol{\medcircle}{\mathbin}{symbolsC}{7}

\crefname{equation}{}{} 
\AtBeginEnvironment{appendices}{\crefalias{section}{appendix}} 

\usepackage[color,final]{showkeys} 

\colorlet{refkey}{orange!20}
\colorlet{labelkey}{blue!30}

\numberwithin{equation}{section}
\newtheorem{theorem}{Theorem}[section]
\newtheorem{proposition}[theorem]{Proposition}
\newtheorem{lemma}[theorem]{Lemma}

\newtheorem{corollary}[theorem]{Corollary}

\newtheorem*{question*}{Question}

\theoremstyle{definition}
\newtheorem{definition}[theorem]{Definition}

\newtheorem*{definition*}{Definition}
\newtheorem*{theorem*}{Theorem}

\newtheorem{remark}[theorem]{Remark}

\newcommand{\mb}{\mathbb}
\newcommand{\mbf}{\mathbf}

\newcommand{\mc}{\mathcal}

\let\originalleft\left
\let\originalright\right
\renewcommand{\left}{\mathopen{}\mathclose\bgroup\originalleft}
\renewcommand{\right}{\aftergroup\egroup\originalright}

\renewcommand{\Re}{\operatorname{Re}}
\renewcommand{\Im}{\operatorname{Im}}

\allowdisplaybreaks

\title{Classical and Quantum Speedups for Non-Convex Optimization via Energy Conserving Descent}
\date{}

\author{%
  Yihang Sun\thanks{Equal contribution.}\textsuperscript{\phantom{*},}\thanks{Department of Management Science 
and Engineering, Stanford University. \texttt{kimisun@stanford.edu}}
  \and
  Huaijin Wang\footnotemark[1]\textsuperscript{\phantom{*},}\thanks{Leinweber Institute for Theoretical Physics, Stanford University. \texttt{jeanhjw@stanford.edu}}
  \and
  Patrick Hayden\thanks{Leinweber Institute for Theoretical Physics, Stanford University and Google DeepMind.}
  \and
  Jose Blanchet\thanks{Department of Management Science 
and Engineering, Stanford University.}
}

\begin{document}
\maketitle
\begin{abstract}
We present the first analytical study of ECD, focusing on the one-dimensional setting for this first installment. We
formalize a stochastic ECD dynamics (sECD) with energy-preserving noise, as well as a quantum analog of the ECD
Hamiltonian (qECD), providing the foundation
for a quantum algorithm through Hamiltonian simulation in a tractable model where the barrier-crossing mechanism can be computed explicitly. For one-dimensional double-well objectives in the under-guessing regime, we compute the expected dynamical hitting times from a local minimum to the global minimum. 
We prove that both sECD and qECD exhibit exponential improvements in continuous hitting time relative to their respective gradient-based baselines, stochastic gradient descent (SGD) and quantum tunneling walk (QTW). For objectives with tall barriers, qECD admits a further hitting time improvement over sECD.
Mechanistically, ECD sidesteps the exponential cost associated with rare-escape events of SGD from local minima by moving from dissipative to energy-conserving dynamics.
\end{abstract}


\section{Introduction}
\label{sec:intro}
\begin{table*}[t]
\centering
\caption{We study general one-dimensional positive double-well objectives, but instantiate our results to the symmetric case $V(\Theta)={\omega^2}(\Theta^2-a^2)^2/{8a^2}+V_0$ with $V_0>0$. Here, $\beta \coloneqq V(0)=a^2\omega^2/8$ is the barrier height (see figure in \cref{sec:fig}). We compare hitting times from local minimum at $-a$ to global minimum at $+a$ for standard gradient descent methods (SGD \citet{sgd} and QTW \citet{qtw}) and for sECD and qECD dynamics. The parameters $s, h, E, \lambda_c, \lambda_q$ are tunable in the respective dynamics and treated here as constants. Our results are asymptotic as $\beta\to\infty$.}
\label{tab:mainresults}
\begin{tabular}{llcc}
\toprule
& & \textbf{Classical} & \textbf{Quantum}\\
\midrule
\textbf{Gradient Descent Methods} &  & $\asymp \frac{\sqrt{s}}{a\omega^3}\exp(\omega^2a^2/s)$ & $\asymp \frac{1}{a\omega^{3/2}\sqrt{h}}\exp(a^2\omega/h)$ \\
\midrule
\multirow{2}{*}{\textbf{Energy Conserving Descent}} & $V_0 \gtrsim \beta :$ & $\asymp \frac{\lambda_c a^{3/2}V_0^{1/4}}{\sqrt{\omega}} + \frac{a\sqrt{E}}{\sqrt{V_0}}$ & $ \lesssim \frac{\lambda_q a^2}{V_0}$ \\
& $ V_0 \lesssim \beta:$ & $ \asymp \left(\lambda_c a^2+\frac{\sqrt{E}}{\omega}\right)\log\left(\frac{\beta}{V_0}\right)$  & $ \lesssim \frac{\lambda_q}{ \omega^2}\log^2\left(\frac{\beta}{V_0}\right)$ \\
\bottomrule
\end{tabular}
\end{table*}
\subsection{Background}
Gradient descent methods are the dominant approach for large-scale non-convex optimization due to their scalability. The recently proposed \emph{Energy Conserving Descent} (ECD) framework \citep{de2022born, deluca2023improvingenergyconservingdescent} has been reported to achieve competitive---and in some settings improved---performance relative to widely used gradient-descent-based optimizers such as stochastic gradient descent (SGD), Adam \citep{adam} and SGD with momentum (SGDM) \citep{pmlr-v28-sutskever13}. These results are given across a range of small to medium-scale benchmarks, including ImageNet-1K fine-tuning and Open Graph Benchmark node-classification tasks.

ECD is a physics-inspired dynamical system governed by a Hamiltonian with a purely kinetic term, where the particle's position-dependent mass is inversely proportional to the potential function that encodes the objective. ECD preserves an energy invariant, which prevents the dynamics from converging to strict local minima in the absence of explicit stopping criteria \citep{deluca2023improvingenergyconservingdescent}.
In contrast, rigorous guarantees for SGD to avoid or escape strict local minima typically require additional structural assumptions, such as smoothing or one-point convexity \citep{kleinberg2018alternativeviewdoessgd}. Moreover, in small-stepsize regimes relative to the height of escape barriers, diffusion-approximation theory predicts that escape times from local minima scale exponentially in an inverse-noise or inverse-stepsize parameter \citep{hu2019, mori2022}.

{This frames gradient-based escape in the classical metastability picture: for Langevin-type dynamics and diffusion approximations of SGD, Kramers--Eyring--Kramers and Freidlin--Wentzell theory predict rare-event escape times with exponential dependence on a barrier-to-noise ratio \citep{kramers1940brownian,freidlin2012random,berglund2013kramers,hu2019,mori2022}. The sECD noise has a different role: it randomizes direction while preserving the energy invariant, so barrier crossing in the under-guessing regime is not powered by a rare escape event.}

The existence of a Hamiltonian description also motivates a natural extension of the classical optimization dynamics into a quantum problem. An analogous quantization of SGD called \emph{quantum tunneling walk} (QTW) has previously been studied where a quantum advantage over SGD driven by quantum tunneling effects has been observed when the objective function has tall barriers \citep{qtw}. 

Our contribution is the first formulation and analytical study of a continuous stochastic ECD dynamics (sECD) and a quantum analog of the ECD Hamiltonian (qECD). We evaluate them through expected hitting times from a local minimum to the global minimum on {one-dimensional double-well potentials in the under-guessing regime. We observe exponential improvements in hitting time scaling relative to the respective baselines SGD and QTW}. For tall barriers, qECD has a further hitting time improvement over sECD in our one-dimensional model.
{While a query-complexity version of the statement is beyond the scope of this paper, we emphasize the mechanistic improvement of ECD: by changing the underlying motion to be energy-conserving, ECD in the under-guessing regime sidesteps the exponential cost of the rare escape mechanism of dissipative algorithms like SGD; it is not a speedup of the same mechanism.}

\subsection{Main Results and Organization}
\label{subsec:results}

In non-convex optimization, we seek to minimize an objective function \(F(\Theta)\) over \(\Theta \in \mb{R}^d\). ECD requires an a priori guess \(F_0\) for the global minimum \(\min F\), which defines the potential
$
V(\Theta)\coloneqq F(\Theta)-F_0,
$
and simulates a classical or quantum dynamics that conserves a total energy $E$ (classical) or an energy distribution $\omega(E)$ (quantum).

In \cref{subsec:det-ecd}, we describe the qualitative regimes of ECD determined by comparing $F_0$ and $\min F$. This work focuses on one-dimensional double-well objectives in the \emph{under-guessing} regime $F_0<\min F$, equivalently $V_0\coloneqq \min F-F_0>0$. Our primary performance metric is the expected hitting time for the dynamics from the local minimum to the global minimum.

In \cref{subsec:add-noise,subsec:qecd}, we formalize the sECD and qECD dynamics:
to make sECD a viable optimizer, we introduce energy-preserving noise with tunable rate $\lambda_c >0$ and hitting time is defined for this dynamics; in the quantum setting where continuous monitoring is unavailable without algorithmic overhead, the hitting time is defined via a quantum-walk-like protocol. Although Planck's constant $\hbar$ sets an overall energy scale for the semiclassical limit, we isolate the dynamical structure of the Hamiltonian via nondimensionalization, so that the dependence on physical constants is absorbed into a tunable rate $\lambda_q>0$ analogous to $\lambda_c$.

In \cref{sec:classical,sec:quantum}, we compute expected hitting times for sECD and qECD on general one-dimensional positive double-well potentials. Our assumptions on $V$ and the main results are stated in \cref{subsec:assumptions} and \cref{thm:general-classical,thm:general-quantum}. 
For concreteness, \cref{tab:mainresults} instantiates these results for
\begin{equation}\label{eq:doublewell}
    V(\Theta)=\frac{\omega^2}{8a^2}(\Theta^2-a^2)^2+V_0,\qquad V_0>0.
\end{equation}
The tunable parameters in ECD are the energy $E$ (classical) and the learning rates $\lambda_c$ (classical) and $\lambda_q$ (quantum); the latter are morally analogous to learning rates $s$ and $h$ in SGD and QTW. {The notation table is given in \cref{tab:notation}.} Let $\beta\coloneqq V(0)=a^2\omega^2/8$ denote the barrier height; \citet{qtw} observes a quantum improvement in expected hitting time for QTW over SGD as $\beta\to\infty$.

We compare the expected time for dynamics initialized at $-a$ to reach $+a$, separated by how under-guessing error $V_0$ compares to barrier height $\beta$. In both regimes, we show that
\begin{itemize}
      \item {sECD and qECD exhibit exponential improvements in continuous
    hitting time scaling relative to their respective baselines,
    SGD and QTW. This reflects a change in escape mechanism:
    gradient-based dynamics follow a rare-event barrier-crossing
    \citep{freidlin2012random,kramers1940brownian}, while ECD crosses barriers through energy-conserving
    motion with directional randomization or quantum propagation.}
    \item As conserved energy $E$ is a tunable parameter for sECD but not for qECD, we compare qECD to the energy-independent lower bound on the sECD hitting time. As barrier height $\beta\to\infty$, qECD attains a further $\Omega(\beta/\log \beta)$ {asymptotic hitting time} improvement over sECD, mirroring the separation between QTW and SGD in \citet{qtw}.
\end{itemize}
We expand and justify these comparisons in \cref{sec:comparison}, {including experimental validation for \cref{tab:mainresults} with code found in \href{https://github.com/jean-wang-huaijin/Classical-and-Quantum-ECD}{this repository}}. In \cref{sec:summary}, we conclude with a discussion of forthcoming work and future directions.

\subsection*{Acknowledgments}
We thank Hideo Mabuchi and Eva Silverstein for helpful discussions. We thank the anonymous reviewers and the area chair for
their insightful feedback. YS is funded by the NSF Graduate Research Fellowship and the Stanford Graduate Fellowship. HW acknowledges support from the Stanford QFARM Quantum Science Seed Grant and Knight-Hennessy Scholars. 
JB acknowledges support from ONR under award N00014-24-1-2655, and the National Science Foundation (NSF) under grants 2312204 and 2403008.

\section{Preliminaries}
\label{sec:prelim}

\subsection{Deterministic ECD}\label{subsec:det-ecd}
The deterministic ECD algorithm in $\mb{R}^d$ was proposed by \citet{de2022born,deluca2023improvingenergyconservingdescent,de2025optimizers} as the first-order discretization of the coupled differential equations of \emph{position} $\Theta_t$ and \emph{momentum} $\Pi_t$, i.e.
\begin{equation}
\label{eq:ecd-cont-det}
\frac{d\Theta_t}{dt}=\frac{2\Pi_t}{\Vert \Pi_t\Vert^2}\quad\text{and}\quad \frac{d\Pi_t}{dt}=-\frac{\nabla V(\Theta_t)}{V(\Theta_t)}.
\end{equation}
The equations have a time-independent Lagrangian formulation and, therefore, a conserved energy, giving ECD its name. We view it as a tunable parameter $E$.
\begin{lemma}[\citet{deluca2023improvingenergyconservingdescent}]\label{lem:energy-conserved}
Let $E\coloneqq \Vert \Pi_0\Vert^2 V(\Theta_0)$. Then,
$E=\Vert \Pi_t\Vert^2 V(\Theta_t)$ for all $t\ge 0$.
\end{lemma}
We decouple the momentum $\Pi_t$: its norm 
is given by
\begin{equation}\label{eq:def-p}
\Vert \Pi_t\Vert = p(\Theta_t)\quad\text{where}\quad p(\theta) \coloneqq \sqrt{\frac{E}{V(\theta)}}>0.
\end{equation}
Hence, for direction $u_t$ of $\Pi_t$, \cref{eq:ecd-cont-det} simplifies to
\begin{equation}\label{eq:det-ecd-t}
\frac{d\Theta_t}{dt}=\frac{2u_t}{p(\Theta_t)}\quad\text{where}\quad u_t\coloneqq\frac{\Pi_t}{\Vert \Pi_t\Vert}\in\mb{S}^{d-1}
\end{equation}

To define the algorithm, we set $V(\Theta)=F(\Theta)-F_0$ where $F(\Theta)$ is the objective function to minimize and $F_0$ is the guess for the true minimum $\min F$ of $F$.
We qualitatively describe the three regimes, depending on how they compare:
\begin{itemize}
    \item \emph{Exact-Guessing:} if $F_0=\min F$, then $\min V=0$. As we approach the global minimizer, $V\to 0 $ and $p\to\infty$. Hence, $d\Theta_t/dt\to 0$ by \cref{eq:def-p} and we slow to a stop. 
    \item \emph{Over-Guessing:} if $F_0>\min F$, then the same slow-down mechanism means we approach some $\Theta$ where $V(\Theta)=0$ and so $F(\Theta)=F_0$, not the true minimizer.
    \item \emph{Under-Guessing:} if $F_0<\min F$, then $V$ is bounded above zero everywhere, so $p$ is bounded above. Hence, $\Theta_t$ never stationary, and the dynamics is recurrent.
\end{itemize}
ECD behaves fundamentally differently than gradient descent methods like SGD: ECD slows down due to a diverging effective mass, while SGD slows down due to energy dissipation. It also has different failure modes: ECD is sensitive to the guess $F_0$, but it does not get stuck in local minima as long as $V$ is positive.
{These mechanisms are counterintuitive and are the source
of the hitting time separations studied below.
In this paper, we focus on the under-guessing regime, which is
the simplest theoretically and already exhibits the ECD
barrier-crossing mechanism.}

\subsection{One-Dimensional Stochastic ECD (sECD)}
\label{subsec:add-noise}
Fix energy $E>0$ and recall the decoupling $\Pi_t= p(\Theta_t)u_t$. In one dimension, direction $u_t\in\{-1, 1\}$. By \cref{lem:energy-conserved}, $\Pi_t$ is never zero, so sign $u_t\equiv u_0$ never flips. In the under-guessing regime, $p(\Theta_t)$ is bounded above, so in fact $\Theta_t$ moves in direction $u_0$ forever and never returns. 

Unlike gradient descent, if ECD starts climbing up a potential barrier, it will continue to do so forever, and at increasing speed; although momentum decreases, effective mass vanishes. This is the negative counterpart of the miraculous slow-down mechanism in the exactly-guessing discussion.

To mitigate this, noise must be introduced while conserving energy. For the discrete algorithm in $\mb{R}^d$, \citet{deluca2023improvingenergyconservingdescent}
proposes a random rotation of momentum after every time-step $t$: pick $\nu >0$, sample $z\sim N(0, I_d)$, and update
\begin{equation}\label{eq:rot-discrete}
\Pi_t \leftarrow \left(\Pi_t+\nu z\right)\frac{\Vert \Pi_t\Vert}{\Vert \Pi_t+\nu z\Vert}
\end{equation}
which is observed empirically to induce turnarounds.
\begin{definition}\label{def:secd-d}
In a forthcoming paper, we derive the continuum limit as time-step $\Delta\to 0$ and $\nu \sim \eta\sqrt{\Delta}$ to get
{\small \begin{equation}\label{eq:sde-u}
\begin{aligned}
du_t &= \left( -\frac{(I_d - u_t u_t^\top) \nabla V(\Theta_t)}{\|\Pi_t\|V(\Theta_t)} - \frac{\eta^2(d-1)}{2\|\Pi_t\|^2} u_t \right) dt 
\\ & \qquad + \frac{\eta}{\|\Pi_t\|}(I_d - u_t u_t^\top)dB_t
\end{aligned}
\end{equation}}
and \cref{eq:det-ecd-t} as the \emph{sECD dynamics in $\mb{R}^d$} for $d>1$.
\end{definition}

In one dimension, however, rotations are sign flips of $u_t\in\{-1, 1\}$ and do not admit a continuous stochastic analog: flipping probability decays to $0$ exponentially quickly as step-size $\Delta\to 0$, and we also see that \cref{eq:sde-u} degenerates to $du_t=0$ for $d=1$. To recover something meaningful in the spirit of \cref{def:secd-1}, we define the time-change
\begin{equation}
\frac{ds}{dt}=\frac{2}{p(\Theta_t)^2}\implies \frac{d\Theta_s}{ds} = u_s\,p(\Theta_s).
\label{eq:time-change}
\end{equation}
{Equivalently, $dt/ds=p(\Theta_s)^2/2=E/(2V(\Theta_s))$. Thus the Poisson clock is homogeneous only in intrinsic $s$-time. In real time $t$, its instantaneous flip intensity is
\[
\lambda_c\frac{ds}{dt}=\frac{2\lambda_c}{p(\Theta_t)^2}=\frac{2\lambda_c V(\Theta_t)}{E}.
\]
All hitting times in \cref{sec:classical} are reported in real time.}
It turns out that this is the intrinsic time of sECD for $d>1$ and that $u_s\in \mb{S}^{d-1}$ is a diffusion with time and momentum independent parameters, so it flips at some constant rate $\lambda_d>0$. Motivated by the intrinsic-time angular randomization in $d>1$, we introduce the following one-dimensional continuous-time model: a Poisson clock of constant rate $\lambda_c$ in $s$-time flips $u_s$. {We do not claim here that this Poisson-flip process is the weak limit of any discrete one-dimensional algorithm; rather, it is a tractable proposed model that isolates the turnaround mechanism.}
\begin{definition}\label{def:secd-1}
The \emph{one-dimensional sECD dynamics} is \cref{eq:time-change} where $u_s\coloneqq u_0(-1)^{P_s}$ for a Poisson process $P_s$ with constant rate $\lambda_c  >0$ in $s$-time, and initial conditions $(\Theta_0, u_0)$.
\end{definition}
Our goal is to analyze the escape time of sECD from a local minimum to the global minimum in the one-dimensional under-guessing regime for double-well potential $V$ satisfying the assumptions in \cref{subsec:assumptions}. We initialize $\Theta_0=-a $ at a local minimum, and compute the expected hitting time 
\begin{equation}
T_{\mathrm{hit}}\coloneqq \inf\{t>0:\Theta_t=a\}
\end{equation}
in real time to the global minimum at $a$.

\subsection{Quantum ECD (qECD)}
\label{subsec:qecd}
For the quantum dynamics, we employ Dirac notation. A (pure) quantum state is a unit vector $\ket{\phi}$ in a complex Hilbert space $\mathcal{H}=L^2(\mathbb{R}^d)$, with a dual vector $\bra{\phi}=(\ket{\phi})^{\dagger}$ and inner product $\braket{\psi}{\phi}$. A Hamiltonian $H$ is a self-adjoint operator that generates unitary time evolution of the quantum system. Over $\mc{H}$, position basis $\{\ket{x}:x\in \mathbb{R}^d\}$ yields the coordinate representation of the wavefunction $\phi(x):= \braket{x}{\phi}$. 

The quantum ECD dynamics simulates the Schr\"odinger equation: for imaginary unit $i$ and Planck's constant $\hbar$
\begin{equation}
    i\hbar\frac{\partial}{\partial t}\ket{\Phi(t)}=H\ket{\Phi(t)},
\end{equation}
$H$ is the Hamiltonian describing the quantized ECD dynamics. 
There is some freedom in defining a quantum mechanical version of classical deterministic ECD, which requires promoting the classical variables to quantum operators while preserving self-adjointness. In this work, we consider the simplest quantization, addressing the ambiguity in operator ordering by using the symmetric ordering.

\begin{definition}
\label{def:qECD}
    Given a potential function $V(\Theta)$, starting from an initial quantum state $\psi_0(\Theta)$, the \emph{one-dimensional qECD dynamics} is a continuous simulation of the unitary evolution $U(t) = e^{-iHt}$ with the Hamiltonian $H$ given by
    \begin{equation}
    H=-\hbar^2\partial_{\Theta}(V(  \Theta)\partial_{\Theta}).
    \end{equation}
\end{definition}

To employ standard semiclassical analysis, we perform asymptotic expansion as $\hbar \rightarrow 0$. {To compare the optimization dynamics rather than physical units, we introduce a dimensionless time-scale parameter $\lambda_q >0$ in analogy with $\lambda_c$ in sECD via} 
\begin{equation}\label{eq:Hamiltonian}
 \tilde{H} \coloneqq \frac{1}{\lambda_q^2\hbar^2}H=-\frac{1}{\lambda_q^2}\partial_x(V(x)\partial_x),
\end{equation}
and compare the expected hitting time for $\tilde{H}$ with that of sECD. {This rescaling is interpreted only as a choice of hitting time units in the present paper. We do not claim a query-complexity separation here. See further discussions in \cref{sec:summary}.} We note the following features in qECD, which introduces further complexity in performance analysis.

While sECD initializes with position $\Theta_0$ and conserved energy $E$, qECD initializes with a quantum wavefunction $\psi_0(\Theta)$ which prescribes an energy spectral measure $\omega(E)$. This subtle difference is further discussed in \cref{subsec:comparison-setup}.

Since the quantum algorithm cannot be continuously monitored without affecting the state of the system, we choose the standard randomized protocol from quantum walks literature \citep{10.1145/780542.780552}, which evolves the quantum system for a randomly chosen time $t \in [0,\tau]$. At time $t$, we denote the amplitude of the quantum state by $\psi(\Theta,t)$, so the probability density of the wavefunction is $\Theta\mapsto |\psi(\Theta,t)|^2$.
Then, the probability of successfully detecting the state within a $\sigma$-neighborhood around the global minimum $a$ at time $t$ is
\begin{equation}\label{eq:p-eps}
p_{\sigma}(t)=\int_{a-\sigma}^{a+\sigma} d\theta |\psi(\theta,t)|^2,
\end{equation}
for some small $\sigma =O(\sqrt{\hbar})$ we choose. Physically, this is a natural choice of length scale set by the ground state of a harmonic potential. This protocol is then repeated independently to ensure at least one successful detection. Based on this protocol, we introduce the notion of hitting time for comparison with classical results.
\begin{definition}
\label{def:hitting}
The \emph{expected hitting time} of the qECD dynamics starting from $\psi_0(\Theta)$ to a $\sigma$-neighborhood of $a$ is 
    \begin{equation}\label{eq:T-eps}
        T_{\epsilon}(\psi_0|a)\coloneqq \inf_{\tau > 0 }\frac{\tau}{\overline{p_{\sigma}}(\tau)}\quad\text{where}\quad \overline{p_{\sigma}}(\tau)\coloneqq\int_0^{\tau}\frac{dt}{\tau}p_{\sigma}(t)
    \end{equation}
    is the averaged success probability for one trial.
\end{definition}

\subsection{Assumptions and Notations}\label{subsec:assumptions}
We focus on a one-dimensional double-well objective $F:\mb{R}\to \mb{R}$ in the under-guessing regime ($F_0<\min F$) and make two assumptions on the potential function $V = F-F_0$:
\begin{itemize}
\item $V$ is a $C^1$-function with exactly two local minima which we assume without loss of generality to be $\{-a, a\}$ for $a>0$. Let $V_0\coloneqq V(a) >0$ be the global minimum and let $V_1\coloneqq V(-a) >V_0$ be the local minimum.
\item Tail condition: $V(\theta)\to\infty$ as $|\theta|\to\infty$, and
\begin{equation}
\int_{|\theta|>a}\frac{d\theta}{\sqrt{V(\theta)}}<\infty.
\label{eq:tail-int}
\end{equation}
A sufficient condition is super-quadratic growth: there exists $c >2$ such that
$V(\theta)\gtrsim |\theta|^{c}$ as $|\theta|\to\infty$.
\end{itemize}

We adopt standard asymptotic notation, with the limiting regime clear from context:
$f \lesssim g$ denotes $f=O(g)$,
$f \ll g$ denotes $f=o(g)$,
$f \asymp g$ denotes $f=\Theta(g)$,
and $f \sim g$ denotes $f/g = 1+o(1)$.
For a Markov process $X_t$, we adopt standard notation $\mb{E}_x$ and $\mb{P}_x$ as expectation and probability conditioned on $X_0=x$.

\section{Analyzing Stochastic ECD}
\label{sec:classical}
\subsection{Setup and Main Result}
\label{subsec:setup}
We begin with a coordinate-change: let $x_t\coloneqq \phi(\Theta_t)$ where
\begin{equation}\label{eq:change-x}
\phi(\theta)\coloneqq \int_{-a}^{\theta} \frac{d\xi}{p(\xi)}
=\frac{1}{\sqrt{E}}\int_{-a}^{\theta}\sqrt{V(\xi)}\,d\xi.
\end{equation}
Then, local minima $\theta=\pm a$ map to $x=0$ and $x=L\coloneqq \phi(a)$.
By \cref{eq:time-change}, the $x$-dynamics is
${dx_s}/{ds}=u_s\in\{\pm1\}$,
and for any $s$-time stopping $S$, the corresponding real time is
\begin{equation}\label{eq:T-S-conversion}
T(S)\coloneqq \frac{1}{2}\int_0^Sp(\phi^{-1}(x_s))^2\,ds.
\end{equation}
We can compute the hitting time for the noiseless dynamics.
\begin{proposition}
\label{prop:Tdet}
If $u_s\equiv 1$ is constant, and initially $\Theta_0=-a $, then the deterministic real time to hit $a$ is
\begin{equation}
T_{\det}
= \frac{1}{2}\int_0^L p(\phi^{-1}(x))^2\,dx
= \frac{1}{2}\int_{-a}^{a} p(\theta)\,d\theta.
\label{eq:Tdet}
\end{equation}
\end{proposition}
Our main result for the classical case is the following analog for double-wells when we inject noise as in \cref{subsec:add-noise}.
\begin{theorem}
\label{thm:general-classical}
For $V$ satisfying assumptions in \cref{subsec:assumptions}, and the stochastic ECD (sECD) dynamics in \cref{def:secd-1} with $\Theta_0=-a$, the expected hitting time to $a$ is given by
\begin{equation}
\begin{aligned}
& T_{\det}+\lambda_c  \int_{-a }^{a}\left(\int_{\theta}^{a}\frac{d\xi}{p(\xi)}\right)p(\theta)\,d\theta
\\& \quad +\left(\lambda_c  L+\mbf{1}_{\{u_0=-1\}}\right)\int_{-\infty}^{-a} p(\theta)d\theta
\end{aligned}
\end{equation}
with $T_{\det}$ from \cref{eq:Tdet}, $L=\phi(a)$ from \cref{eq:change-x}, and $p$ from \cref{eq:def-p}.
\end{theorem}
We observe that the first line corresponds to the time spent exploring $[-a , a]$ and it converges to the deterministic hitting time as noise $\lambda_c \to 0$; the second line corresponds to time spent exploring the tail $[-\infty, -a ]$, and there is an extra term if we initialize towards the tail, i.e. $u_0=-1$.

We simplify if $V$ is \emph{symmetric}, i.e. $V(\xi)=V(-\xi)$ for all $\xi$.
\begin{corollary}
\label{cor:symm-case}
If $V$ is a symmetric double-well, then the expected hitting time to $a$ of sECD with $\Theta_0=-a$ is
\begin{equation}
(1+\lambda_c  L)\int_{0}^\infty p(\theta)d\theta-\mbf{1}_{\{u_0=1\}}\int_{a}^\infty p(\theta)d\theta.
\end{equation}
\end{corollary}

\subsection{Embedded Four-State Markov Chain}
\label{subsec:embedded}
Under the space-change in \cref{eq:change-x}, let $S_n$ be the $s$-time of the $n$-th visit to minima $\{-a ,a\}$, and let
$Z_n\coloneqq (X_n,U_n)\coloneqq (x_{S_n},u_{S_n})$. Then $(Z_n)$ is a Markov chain on four states
\begin{equation}
\label{eq:4-states}
\mc{Z}\coloneqq \{(0,-),(0,+),(L,-), (L,+)\}.
\end{equation}
We call $(0,+)$ and $(L,-)$ \emph{inward states}, while $(0,-)$ and $(L,+)$ are \emph{outward states}, relative to the interval $[0, L]$. 

The transitions $(0,-)\to(0,+)$ and $(L,+)\to(L,-)$ from outward states are deterministic: we must return to the inward state at the same well before hitting the other well. 
Starting at an inward state, we could either cross the barrier and preserve signs, or return to the starting position and reverse signs. Let $q$ be the crossing probability. These notions are spelled out in \cref{subsec:telegraph} via the telegraph process.

The transition matrix with states listed in order as \cref{eq:4-states} is
\begin{equation}
P\coloneqq 
\begin{pmatrix}
0 & 1 & 0 & 0\\
1-q & 0 & 0 & q\\
q & 0 & 0 & 1-q\\
0 & 0 & 1 & 0
\end{pmatrix}.
\label{eq:P-matrix}
\end{equation}
Let $D_{\mathrm{hit}}$ be the number of discrete steps until first arrival to the global well $a$, i.e. the smallest $n$ such that $x_n=L$. Without computing $q$, we make the following observation.

\begin{proposition}
\label{prop:legs-hit}
The stationary distribution of $P$ is uniform on $\mc{Z}$, and the expected hitting time is
\begin{equation}
\mb E_{(0,+)}[D_{\mathrm{hit}}]=\frac{2}{q}-1,\qquad \mb E_{(0, -)}[D_{\mathrm{hit}}]=\frac{2}{q}.
\label{eq:legs-hit}
\end{equation}
\end{proposition}

\subsection{The Telegraph Process}
\label{subsec:telegraph}

We now isolate and analyze the key process governing $q$ and barrier crossing, which will allow us to compute transition times of $Z_n$ and hitting times of $\Theta_t$ and $x_s$.

Let $(X_s,U_s)$ be the telegraph process with $X_0=0$, $U_0=1$, and $dX_s/ds=U_s$ which flips with rate $\lambda_c  >0$. Define
\begin{equation}
S\coloneqq \inf\{s>0:\ X_s\in\{0,L\}\}.
\label{eq:exit-time}
\end{equation}
as the exit time on $[0, L]$. 
All results below are proved by solving the corresponding boundary value problems
$\mathcal L h=0$ (harmonic) or $\mathcal L g=-w$ (Poisson) for generator $\mc{L}$. They are standard exercises deferred to \cref{sec:telegraph-proofs}.

\begin{proposition}
\label{prop:q}
Define $q\coloneqq \mb P_{(0,+)}(X_{S}=L)$. Then,
\begin{equation}
q=\frac{1}{1+\lambda_c  L}.
\label{eq:q-formula}
\end{equation}
\end{proposition}

For an integrable running cost $w:\mb{R}\to[0,\infty)$ define
\begin{equation}
G_\pm(x)\coloneqq \mb E_{(x,\pm)}\left[\int_0^{S} w(X_s)\,ds\right].
\end{equation}

\begin{lemma}
\label{lem:poisson-in}
Assume $G_-(0)=0$ and $G_+(L)=0$. Then
\begin{equation}
\begin{aligned}
G_+(0)&=q\int_0^L \left(1+2\lambda_c (L-x)\right)w(x)\, dx,\\
G_-(L)&=q\int_0^L \left(1+2\lambda_c  x\right)w(x)\, dx.
\end{aligned}
\label{eq:mu-in-general}
\end{equation}
\end{lemma}

We also need the analogous statement on a half-line $\mb{R}_\pm$ instead of an interval. The computation follows the interval $[0, L]$ case upon taking $L\to\infty$. 
We record this formally.
\begin{lemma}
\label{lem:poisson-out}
For stopping time
$\sigma\coloneqq \inf\{s>0:\ X_s=0\}$ and integrable $w$, then $\sigma <\infty$ almost surely and
\begin{equation}
\mb E_{(x, \pm)}\left[\int_0^{\sigma} w(X_s)\,ds\right] =2\int_{\mb{R}_\pm} w(x)\,dx.
\label{eq:mu-out-halfline}
\end{equation}
\end{lemma}
We will apply this to the ECD running cost $w$ to the tails of $V$, so the tail condition gives integrability of $w$. 

\subsection{Transition and Hitting Times}
\label{subsec:semi-markov}
Consider the semi-Markov process given by $Z_n=(X_n, U_n)$ and real time $T_n$ associated with $Z_n$. Let $\Delta t$ be the real time associated with one transition step of $Z_n$, analogous to $T_{\det}$ from \cref{eq:Tdet}. We compute $\Delta t$ by \cref{lem:poisson-in,lem:poisson-out} with
\begin{equation}
\label{eq:def-w}
w(x)\coloneqq \frac{dt}{ds}=\frac{1}{2}p(\phi^{-1}(x))^2.
\end{equation}
\begin{proposition}
\label{prop:one-step-t}
For outward states $z\in \{(0, -), (L, +)\}$
\begin{equation}
\begin{aligned}
\mb E_{(0, -)}[\Delta t]&=\int_{-\infty}^{-a }p(\theta)\,d\theta,\\ 
\mb E_{(L, +)}[\Delta t]&=\int_{a}^{\infty}p(\theta)\,d\theta,
\end{aligned}
\end{equation}
are finite by \cref{eq:tail-int}.
For an inward state $z\in \{(0, +), (L, -)\}$
\begin{equation}
\mb E_{z}[\Delta t]=\left({T_{\det}+\lambda_c  B_z}\right)q,
\end{equation}
where $q$ and $T_{\det}$ are defined in \cref{eq:q-formula,eq:Tdet}, as well as
\begin{equation}
\label{eq:Bz}
\begin{aligned}
B_{(0, +)} &\coloneqq \int_{-a }^{a}\left(\int_{\theta}^{a}\frac{d\xi}{p(\xi)}\right)p(\theta)\,d\theta,\\
B_{(L, -)}
&\coloneqq \int_{-a }^{a}\left(\int_{-a }^{\theta}\frac{d\xi}{p(\xi)}\right)p(\theta)\,d\theta.
 \end{aligned}
\end{equation}
\end{proposition}

Let $T_{\mathrm{hit}}$ and $S_{\mathrm{hit}}$ be the real time and $s$-time to first hit the global well $a$, analogous to $D_{\mathrm{hit}}$. 
With $q$ in \cref{eq:q-formula}, we now compute their expectations, similar to \cref{prop:legs-hit}.
\begin{proposition}
\label{prop:real-hit}
Under \cref{eq:tail-int}, the stationary distribution of the semi-Markov process is proportional to $\mb{E}_z[\Delta t]$, and
\begin{equation}
\label{eq:real-hit}
\begin{aligned}
\mb E_{(0,+)}[T_{\mathrm{hit}}] & = \frac{\mb{E}_{(0, +)}[\Delta t]+(1-q)\mb{E}_{(0, -)}[\Delta t]}{q}, \\ 
\mb E_{(0,-)}[T_{\mathrm{hit}}] & = \frac{\mb{E}_{(0, +)}[\Delta t]+\mb{E}_{(0, -)}[\Delta t]}{q}.
\end{aligned}
\end{equation}
\end{proposition}
Combining \cref{prop:one-step-t,prop:real-hit} gives \cref{thm:general-classical}. We remark that the key simplification via symmetry in \cref{cor:symm-case} is $B_{(0, +)}=B_{(L, -)}$ and they sum to $2LT_{\det}$.

\section{Analyzing Quantum ECD} \label{sec:quantum}
\subsection{Setup and Main Result}\label{subsec:qsetup}
A quantum state evolving under the qECD dynamics has a position-dependent local momentum $p$ given by \cref{eq:def-p}. Analogous to \cref{eq:change-x}, we define the Liouville coordinates 
\begin{equation}\label{eq:ycoord}
    y(\Theta) \coloneqq \int_{0}^\Theta \frac{dx}{\sqrt{V(x)}}.
\end{equation}
To simplify notation, we also define the distance integral
\begin{equation}\label{eq:I}
    I(\Theta_1,\Theta_2) \coloneqq y(\Theta_2)-y(\Theta_1).
\end{equation}
We fix objective function $V$ satisfying \cref{subsec:assumptions} and take the semiclassical limit $\hbar \to 0$ to compute the expected hitting time (\cref{def:hitting}) from $\psi_0$ to $a$ for qECD and a general positive double-well $V$. 
Before stating our main results on hitting times of qECD, we formalize parameter choices for the qECD dynamics. 
\begin{itemize}
    \item The initial quantum state is a zero-momentum Gaussian wavefunction of width $\sigma = O(\sqrt{\hbar})$ centered at the local minimum $-a$, i.e.
\begin{equation}\label{eq:psi0}
    \psi_0(\Theta) \coloneqq \frac{1}{(2\pi\sigma^2)^{1/4}}\exp{-\frac{(\Theta+a)^2}{4\sigma^2}}.
\end{equation}
\item The detection window $[a-\sigma, a+\sigma]$ around global minimum $a$ has the same width, following \cref{eq:p-eps}. 
\end{itemize}

Heuristically, the choice of $\sigma$ guarantees that as $\hbar \to 0$, the potential function in the respective windows is $O(\hbar)$ away from $V(a)=V_0$ and $V(-a)=V_1$, respectively. 
\begin{theorem}[Quantum Hitting Time]
\label{thm:general-quantum}
For $V$ satisfying assumptions in \cref{subsec:assumptions} and $\sigma  = O(\sqrt{\hbar})$, there exists an absolute numerical constant $c>0$ such that the expected hitting time of the qECD dynamics (\cref{def:qECD}) with Hamiltonian $H$, as $\hbar\to 0$, is at most
\begin{equation}
\frac{I(-a,a)^2}{\hbar}\sqrt{\frac{V_0}{V_1}}\left[c+ O (\hbar)\right].
\end{equation} 
\end{theorem}
As discussed in \cref{subsec:qecd}, we compare the expected hitting time for sECD with that of qECD under the time-evolution of the rescaled $\tilde{H}$ in \cref{eq:Hamiltonian} for sufficiently small $\hbar$ to obtain
\begin{corollary}
\label{cor:general-quantum}
In the setting of \cref{thm:general-quantum}, the expected hitting time for the qECD with rescaled Hamiltonian $\tilde{H}$ is
    \begin{equation}
        T_{hit}(\psi_0|a) \le 2c \lambda_q I(-a,a)^2\sqrt{\frac{V_0}{V_1}}.
    \end{equation}
\end{corollary}

\subsection{Energy-domain Analysis}\label{subsec:qenergy}
The energy eigenstates of the Hamiltonian $H$ can be derived from the time-independent Schr\"odinger equation
{\begin{equation}\label{eq:TISE}
    -\hbar^2(\partial_\Theta V\partial_\Theta)\psi = E\psi.
\end{equation}}
Naively, the absence of "turning points" in the under-guessed landscape suggests a continuous spectrum. Upon closer inspection, the super-quadratic tails assumed in $\cref{subsec:assumptions}$ enforces a discrete, quantized spectrum due to domain compaction in Liouville coordinates. 
\begin{lemma}\label{lem:discrete}
    Under assumptions in \cref{subsec:assumptions}, the qECD Hamiltonian has a discrete energy spectrum bounded below by 0, and a finite effective domain length given by
    \begin{equation}
        L = I(-\infty,\infty)<\infty.
    \end{equation}
\end{lemma}
We separate the low energy (non-semiclassical) and high energy (semiclassical) contributions according to the WKB approximation. Define the WKB energy cut-off
\begin{equation}\label{eq:WKB}
    E_{cut} \coloneqq \hbar^2 \max\left\{\sup_{x \in [-2a, 2a]}\frac{\abs{V'(x)}^2}{V(x)},\sup_{x \in [-2a, 2a]}\abs{V''(x)}\right\}.
\end{equation}
Within the semiclassical regime where $E\gg E_{cut}$, we take the asymptotic expansion of the plane wave ansatz $\psi(\Theta) = \exp[i\phi(\Theta)/\hbar]$ as $\hbar\to 0$ to obtain the following approximate solution to the Schr\"odinger equation \cref{eq:TISE}.

\begin{lemma}\label{lem:wkb}
For $E\gg E_{cut}$,
the qECD Hamiltonian as $\hbar \rightarrow 0$ has quantized energy eigenstates (indexed by $n$)
    \begin{equation}
        \psi_{n}(\Theta) = \frac{\sqrt{2}}{\sqrt{L}V(\Theta)^{1/4}}\sin\left(\frac{n\pi}{L}\int_{-\infty}^{\Theta}\frac{1}{\sqrt{V(x)}}dx +  O (\hbar)\right)
    \end{equation}
\end{lemma}

\subsection{Time-domain Analysis}\label{subsec:qprop}
Given the time-propagator kernel $K(\Theta_2,t;\Theta_1,0)\coloneqq\bra{\Theta_2}e^{-iHt/\hbar}\ket{\Theta_1}$, and an initial wave packet $\psi_0(\Theta)$, its amplitude at time $t$ under the Hamiltonian evolution is
\begin{equation}\label{eq:evolvedstate}
    \psi(\Theta,t) \coloneqq \int d\Theta' K(\Theta,t;\Theta',0)\psi_0(\Theta'),
\end{equation}
and the time-propagator kernel admits expansion
\begin{equation}\label{eq:K}
    K(\Theta_2,t;\Theta_1,0) = \sum_{E_n}e^{-iE_nt/\hbar}\psi_n(\Theta_2)\psi_n^*(\Theta_1).
    \end{equation}
We analyze the \emph{semiclassical} ($E\gg E_{cut}$) and \emph{low-energy} ($E \lesssim E_{cut}$) contributions to $K$ separately as follows. 
\begin{lemma}\label{lem:Klow}
Asymptotically as $\hbar \to 0$, the low-energy contribution to the time-propagator is
\begin{equation}\label{eq:K-low}
    K_{low} \coloneqq \sum_{n:E_n \lesssim E_{cut}}e^{-iE_nt/\hbar}\psi_n(\Theta_2)\overline{\psi_n(\Theta_1)} \lesssim 1.
    \end{equation}
\end{lemma}
\begin{lemma}\label{lem:Kwkb}
There exists a constant $\delta >0 $ depending only on potential $V$ such that asymptotically as $\hbar\to 0$, for any
\begin{equation}\label{eq:time-scale}
    t < \frac{\delta I(\Theta_1,\Theta_2)}{\hbar},
\end{equation} 
the semiclassical contribution to the time-propagator is 
\begin{equation}\label{eq:KWKB}
    \begin{aligned}
     K_{wkb} & \coloneqq \sum_{n:E_n \gg E_{cut}}e^{-iE_nt/\hbar}\psi_n(\Theta_2)\psi_n^*(\Theta_1)\\
    &=  \frac{\exp{i\left[\frac{I(\Theta_1,\Theta_2)^2}{4\hbar t}-\frac{\pi}{4}+ O (\hbar)\right]}}{2\sqrt{\pi\hbar t}(V(\Theta_1)V(\Theta_2))^{1/4}}.
    \end{aligned}
    \end{equation}
\end{lemma}
Up to the time scale \cref{eq:time-scale} that will cover our hitting time (see \cref{lem:qhit}), the phase saddle point lies in the semiclassical region, so the integral in \cref{eq:KWKB} can be evaluated using a quadratic stationary-phase approximation. We now note that $K_{low}$ is asymptotically smaller than $K_{wkb}$ as $\hbar\to 0$.
\begin{corollary}\label{cor:propagatorapprox}
For any time $t$ satisfying \cref{eq:time-scale}, asymptotically as $\hbar \to 0$, the time-propagator of the qECD dynamics is
    \begin{equation}
    K(\Theta_2,t;\Theta_1,0) = \frac{\exp{i\left[\frac{I(\Theta_1,\Theta_2)^2}{4\hbar t}-\frac{\pi}{4}\right]}}{2\sqrt{\pi\hbar t}(V(\Theta_1)V(\Theta_2))^{1/4}}+ O (1).
    \end{equation}
    
\end{corollary}

\subsection{Hitting Time}
\label{subsec:qhitting}
Using the qECD time-propagator, we analyze the time-evolved quantum state \cref{eq:evolvedstate} via a saddle-point expansion of the integral.
\begin{lemma}\label{lem:psit}
    Starting with an initial state \cref{eq:psi0} and evolving under qECD dynamics for time $t$ satisfying \cref{eq:time-scale}, the probability density of the quantum state is given by
    \begin{equation}
    \abs{\psi(\Theta,t)}^2 = \frac{\left[1+ O (\hbar)\right] \exp{\frac{2}{\hbar}\Re{\Phi(\Theta_*')}}}{2t\sqrt{2\pi\alpha^2\hbar V_1 \abs{V(\Theta_*')}}\abs{\Phi''(\Theta_*')}},
     \end{equation}  
where $\alpha = \sigma/\sqrt{\hbar}\lesssim 1$ and $\Theta_*'$ is a complex saddle point of  
    \begin{equation}
    \Phi(\Theta) \coloneqq  -\frac{(\Theta+a)^2}{4\alpha^2}+i\frac{I(\Theta',\Theta)^2}{4t}.
    \end{equation}
\end{lemma}
With measurement window $[a-\sigma,a+\sigma]$, the instantaneous transition probability $p_\sigma$ is given in \cref{eq:p-eps}.
This allows us to derive the following result for its time-averaged counterpart.
\begin{lemma}\label{lem:qhit}
For $t$ satisfying \cref{eq:time-scale} with $(\Theta_1, \Theta_2)=(-a, a)$, the average transition probability $ \overline{p_{\epsilon}}(\tau)$ in \cref{eq:T-eps} is given by 
    \begin{equation}\label{eq:solved-p-eps}
     \overline{p_{\sigma}}(\tau) = \left[1+ O (\hbar)\right]\frac{A}{2\tau}
\int_{{B}/{\tau^2}}^\infty\frac{e^{-z}}{z}dz,
    \end{equation}
asymptotically as $\hbar\to 0$, where
     \begin{equation}
    A \coloneqq  \frac{\sqrt{2}\alpha^2}{\sqrt{\pi V_1 V_0}}\quad\text{and}\quad B \coloneqq \frac{\alpha^2 I_0^2}{\hbar V_1}.
    \end{equation} 
\end{lemma}
Now, we upper bound the expected hitting time in \cref{eq:T-eps} by choosing $\tau$ to be a numerical constant times ${I(-a,a)}/{\sqrt{\hbar V_1}}$. In the appendix, we check that this satisfies \cref{eq:time-scale} and recovers the upper bound in \cref{thm:general-quantum} on the expected hitting time. We remark that we believe a matching lower bound up to numerical constants holds; an analysis of $K_{wkb}$ for time $t$ beyond \cref{eq:time-scale} is required.



\section{Case Study Comparison}
\label{sec:comparison}
\subsection{Setup and Configuration} \label{subsec:comparison-setup}
Recall from \cref{eq:doublewell} the particular symmetric double-well potential in the under-guessing regime, 
corresponding to objective $F(\Theta)={\omega^2}(\Theta^2-a^2)^2/{8a^2}$ and under-guess $F_0=-V_0<0$. 
Note that $V$ satisfies assumptions in 
\cref{subsec:assumptions}. 
We define the \emph{barrier height} of $V$ as
\begin{equation}
\beta\coloneqq V(0)=\frac{1}{8}a^2\omega^2.
\end{equation}
We further specify sECD and qECD configuration beyond \cref{sec:prelim} to enable a fair expected hitting time comparison.
\begin{enumerate}
    \item {In qECD, we assume the initial Gaussian wavepacket $\psi_0$ centered at $-a$ and the detection neighborhood at $+a$ to both have sufficiently small width $\sigma = O(\sqrt{\hbar})$. Our results in the classical setting have no window, but the corresponding expected hitting times for sECD with sufficiently small nonzero initialization and hitting window width $\sigma$ recover $\mb{E}[T_{\mathrm{hit}}]$ asymptotically.}
    \item In qECD, direction $u_0$ and energy $E$ are fixed by $\psi_0$. To guarantee a quantum advantage, we lower bound the classical hitting time by terms independent of $E$ and $u_0$, and show that the qECD hitting time is asymptotically smaller. In particular, this lower bound corresponds to the favorable initialization with $u_0=1$ and $E$ small.
    \item We treat tunable learning rates $s, h, \lambda_c, \lambda_q$ for the respective dynamics as constants in the $\beta\to\infty$ limit.
\end{enumerate}

Next, we justify \cref{tab:mainresults} and conclusions in \cref{subsec:results}. We compute the expected hitting times for $V$ in \cref{eq:doublewell} and compare the resulting dynamical scalings with the SGD and QTW baselines in \citet{sgd,qtw}. Our results in this section are asymptotic as $\beta\to\infty$ and separate into cases according to how the under-guessing error $V_0$ compares with $\beta$.

{The comparison is best read mechanistically: the SGD and QTW baselines are associated with the original objective-landscape geometry, while ECD uses the energy-conserving geometry. Thus the separation is not a smaller Arrhenius prefactor, but a change from rare-event escape to barrier traversal under a different dynamics.} 
\subsection{Small Under-Guessing Error: $V_0\lesssim \beta$}
\label{subsec:small-error}
We compute via \cref{cor:symm-case} and \cref{thm:general-quantum} polynomial expected hitting times of the sECD and qECD dynamics, in contrast with the exponential hitting time scaling of the SGD and QTW baselines in \citet{sgd,qtw}. {This gives the hitting time separation between ECD dynamics and gradient-based escape baselines.}
\begin{theorem}\label{thm:comparison-small}
For potential $V$ in \cref{eq:doublewell} with $V_0\lesssim \beta$, the expected hitting time of sECD is asymptotically
\begin{equation}
T_{c}\coloneqq\left(\lambda_c  a^2+\frac{\sqrt{E}}{\omega}\right)\log\left(\frac{\beta}{V_0}\right)
\end{equation}
for both $u_0\in\{\pm 1\}$; the expected hitting time of qECD is 
\begin{equation}
T_{q}\lesssim \frac{\lambda_q}{ \omega^2}\log^2\left(\frac{\beta}{V_0}\right).
\end{equation}
\end{theorem}
We lower bound $T_c$ uniformly over $E>0$ by its first term to see a $\Omega(\beta/\log \beta)$ factor quantum advantage as $\beta\to\infty$.
\begin{corollary}
Under configuration assumptions in \cref{subsec:comparison-setup}, if $ e^{-O(\beta)} \le V_0\lesssim \beta$ as barrier $\beta\to\infty$, then
\begin{equation}
T_{c}\gtrsim a^2\log\left(\frac{\beta}{V_0}\right) \gg \frac{1}{\omega^2}\log^2\left(\frac{\beta}{V_0}\right) \gtrsim T_q.
\end{equation}
\end{corollary}
The lower bound on under-guessing error $V_0$ ensures the log-term does not dominate; when $V_0$ is exponentially smaller than $\beta$, we are essentially in the exact-guessing regime. Moreover, a natural algorithmic design is to run ECD iteratively in a bisection-style search so $F_0$ converges to $\min F$ from below. Then, quantum speedups persist when $V_0$ is halved for $\Omega(\beta)$ iterations, starting from $V_0\asymp\beta$.

\subsection{Large Under-Guessing Error: $V_0\gtrsim \beta$}
\label{subsec:large-error}
Here, in \cref{cor:symm-case}, the tail integral dominates, so most of the running time classically is exploring $(-\infty, -a]$.
\begin{theorem}\label{thm:comparison-large}
For potential $V$ in \cref{eq:doublewell} with $V_0\gtrsim \beta$, the expected hitting time of sECD is asymptotically
\begin{equation}\label{eq:Tc-large}
T_{c}\coloneqq\left(\mbf{1}_{\{u_0=-1\}}+\lambda_c  a\sqrt{\frac{V_0}{E}}+\frac{\sqrt{a\omega}}{V_0^{1/4}}\right)\sqrt{\frac{aE}{\omega}}V_0^{-1/4},
\end{equation}
depending on $u_0$; the expected hitting time of qECD is 
\begin{equation}
T_{q}\lesssim \frac{\lambda_q a^2}{ V_0}.
\end{equation}
\end{theorem}
We lower bound $T_c$ uniformly over $E>0$ and initial direction $u_0\in \{-1, 1\}$ by its middle term in \cref{eq:Tc-large} to see a $\Omega(\beta)$ factor quantum advantage as barrier height $\beta\to\infty$.
\begin{corollary}
Under configuration assumptions in \cref{subsec:comparison-setup}, if $ V_0\gtrsim \beta$ as barrier height $\beta\to\infty$, then
\begin{equation}
T_{c}\gtrsim a^{3/2}\omega^{-1/2}V_0^{1/4}\gg \frac{a^2}{V_0}\gtrsim T_q.
\end{equation}
\end{corollary}

\subsection{Numerical Simulations}
\label{subsec:numerics}

We numerically check the hitting time scalings in the under-guessing regimes with symmetric quartic double-well potential $V$ with barrier height $\beta\coloneqq V(0)=a^2\omega^2/8$:
\[
V(\Theta)=\frac{{\omega^2}}{8a^2}(\Theta^2-a^2)^2+V_0.
\]
Let classical energy $E=1$, Poisson rate $\lambda_c=0.5$, and initial direction $u_0=1$. We fix $a=1$ and let $\omega\to\infty$ so $\beta\to\infty$, and plot $V_0=1$ for the small $V_0\lesssim\beta$ regime and $V_0=10\beta$ for the large $V_0\gtrsim\beta$ regimes. Both plots have log-scaled axes and are intended to be hitting time comparisons of the dynamics, not query-complexity benchmarks.

\begin{figure}[!htb]
  \centering
  \includegraphics[width=\linewidth]{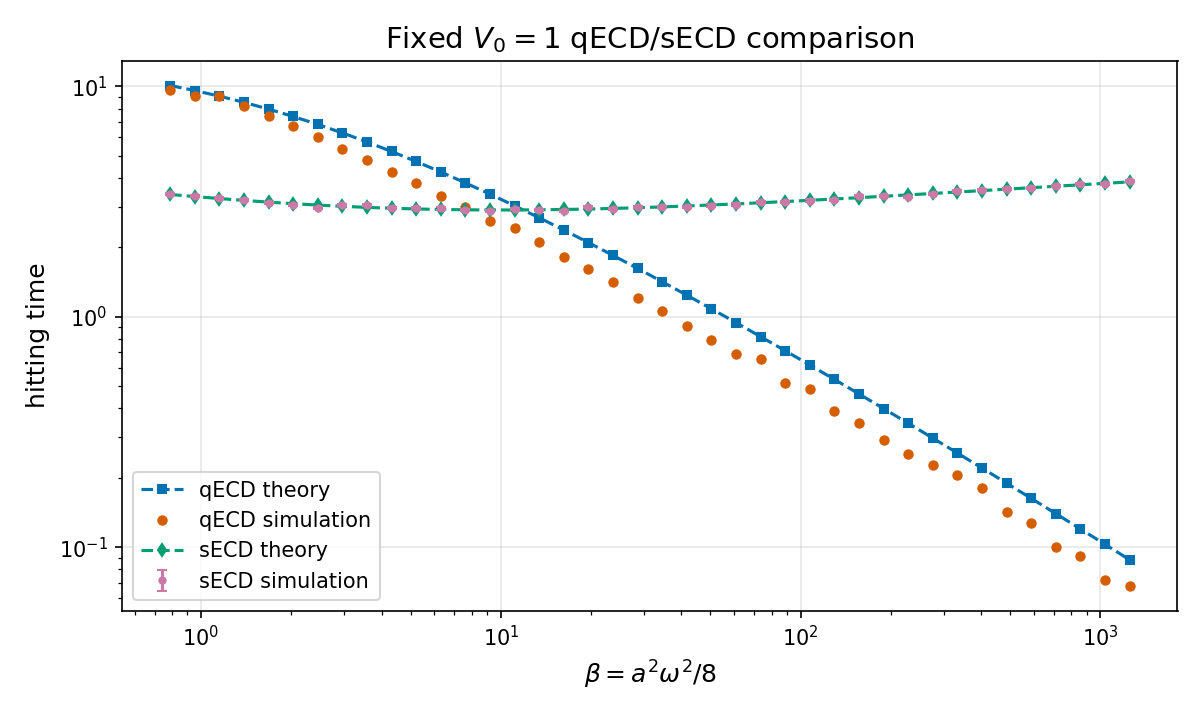}
  \caption{Comparison of qECD and sECD hitting times for the symmetric double-well potential with $V_0=1$.}\label{fixed V0}
\end{figure}

\begin{figure}[!htb]
  \centering
  \includegraphics[width=\linewidth]{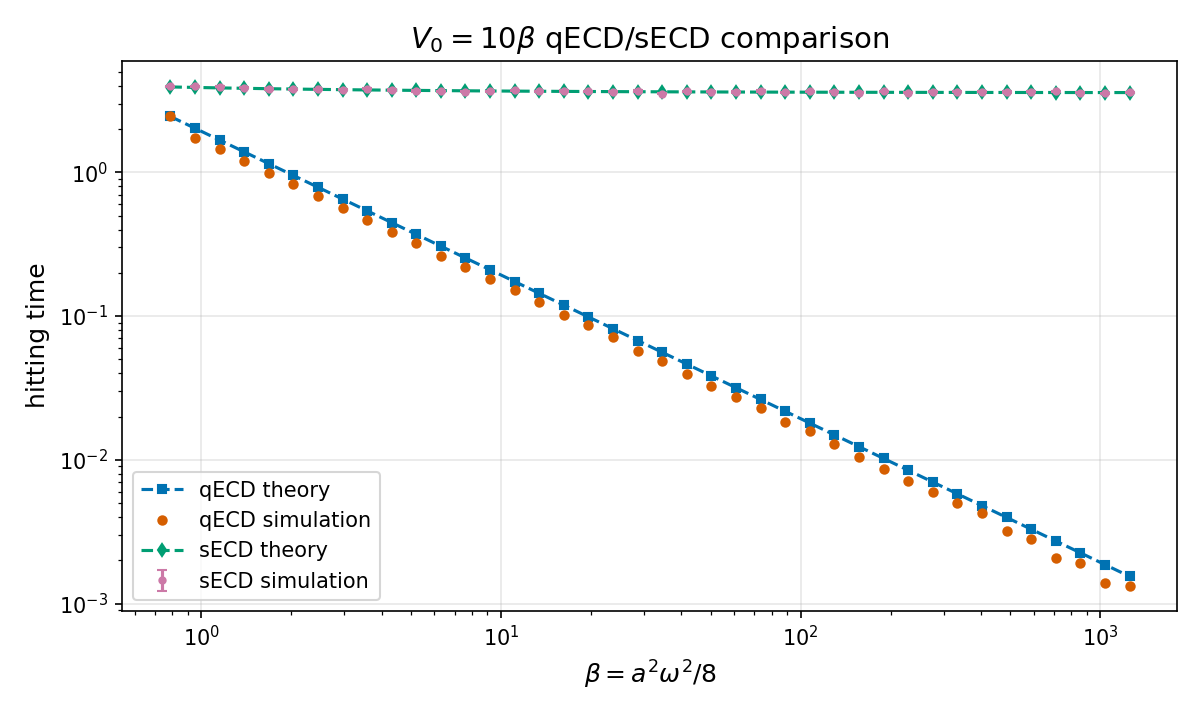}
  \caption{Comparison of qECD and sECD hitting times for the symmetric double-well potential with $V_0=10\beta$.}\label{scaled V0}
\end{figure}

For sECD, we evaluate the exact integral formula in \cref{cor:symm-case} numerically and independently simulate the embedded telegraph process with $8000$ Monte Carlo trajectories. These hitting times match, confirming the validity of \cref{thm:general-classical} in this toy model. For qECD, we discretize $H=-\partial_x(V(x)\partial_x)$, initialize a Gaussian packet of width $\sigma = 0.1/\omega$ centered at the left well, and record the probability of detection in a window with width $\epsilon = \sigma$ around the right well. The reported qECD hitting time proxy is the discrete analogue of \cref{def:hitting}, whereas the theoretical curve is an upper bound from \cref{thm:general-quantum} up to a multiplicative constant. \cref{sec:numerics} describes details on the qECD simulation and \href{https://github.com/jean-wang-huaijin/Classical-and-Quantum-ECD}{this repository} contains the code.

These experiments support the asymptotic hitting time separations as barrier height $\beta\to\infty$: we observe for sECD hitting time the $\Theta(\log\beta)$ growth in \cref{fixed V0} with fixed $V_0=1$ and the plateau in \cref{scaled V0} with $V_0=10\beta$, whereas the qECD hitting time decays to $0$ in both cases, consistent with the theoretical conclusions of \cref{thm:comparison-small,thm:comparison-large}. 

\section{Future Directions}
\label{sec:summary} 

\paragraph{Landscape Generalizations.}
We provide a first analytical proof of concept for ECD by establishing hitting time improvements on one-dimensional double-well objectives
in the under-guessing regime. This is the simplest nontrivial setting in which
the ECD barrier-crossing mechanism and the hitting times can be
computed. Our techniques naturally extend to broader one-dimensional landscapes: multi-well objectives can be coarse-grained into
semi-Markov chains whose edge costs are determined by the Liouville distances
 with real-time conversion following \cref{thm:general-classical}.

A further extension to higher-dimensional objectives via \cref{def:secd-d,def:qECD} is a key direction, as this setting is most relevant for practical applications.
{Classically, \cref{def:secd-d} replaces the one-dimensional Poisson sign flip by angular diffusion on $\mb{S}^{d-1}$, so basin transitions depend on the basin geometry of landscape $V$ rather than a single interval crossing probability. For qECD, the difficulty lies in metric rigidity: the one-dimensional Liouville coordinate used in our analysis has no general global analogue in $d\ge 3$ because Liouville's theorem in conformal geometry restricts conformal coordinate changes, so high-dimensional qECD requires new semiclassical analysis for variable-coefficient Hamiltonians.
Extending this analysis to general landscapes, as well as to other guessing regimes corresponding to sign-indefinite potentials, is the focus of forthcoming work.}

\paragraph{Algorithmic Considerations.}
{A complete notion of algorithmic speedup requires analyzing query complexity and overhead. In forthcoming work, we analyze the speedup by specifying an oracle model for $V$, a discretization and truncation scheme, and error bounds for Hamiltonian simulation. This is then compared with the discretization effects of \cref{def:secd-d} for sECD. We also study robustness to initialization via mixing-time analyses. For sECD, mixing time characterizes convergence to a stationary distribution. For qECD, where unitary dynamics do not converge pointwise, we instead analyze convergence of time-averaged distributions to a dephased limit. This perspective also illustrates the robustness of both ECD dynamics.}



\bibliographystyle{plainnat}
\bibliography{uai2026submission}

\newpage

\appendix
\section{Supplementary Figures}
\label{sec:fig}
In this section, we include an illustration of the double-well potential with parameters given in a notation table.

\begin{figure}[ht]
  \vskip 0.2in
  \begin{center}
\centerline{\includegraphics[width=0.4\columnwidth]{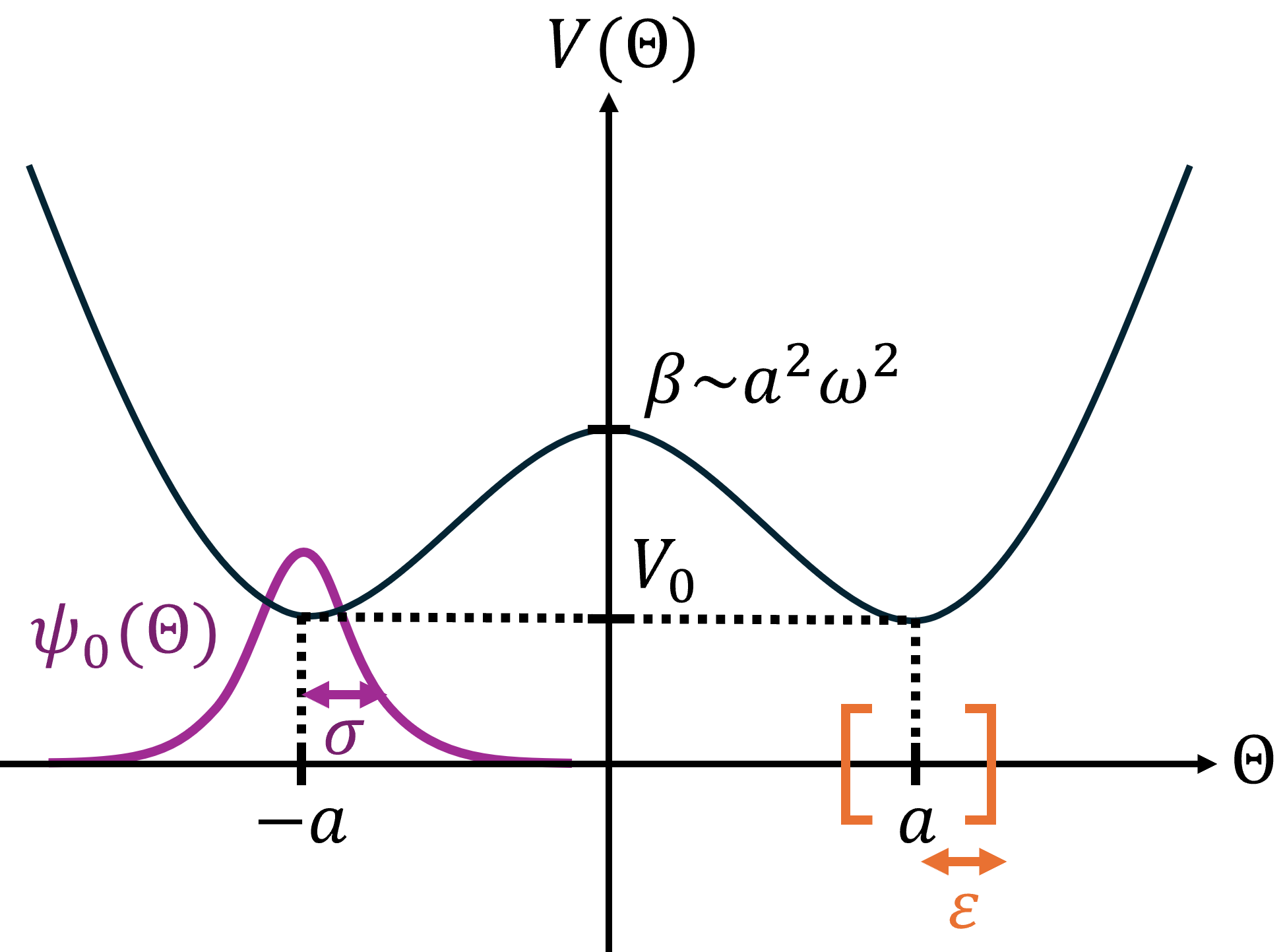}}
    \caption{Symmetric double-well potential in the under-guessing regime.}
    \label{fig:doublewell}
  \end{center}
\end{figure}

\begin{table}[h]
\centering
\small
\caption{{Frequently used notation.}}
\label{tab:notation}
\begin{tabular}{ll}
\toprule
Symbol & Meaning \\
\midrule
$F$ & objective function \\
$F_0$ & scalar guess for $\min F$ \\
$V=F-F_0$ & ECD potential \\
$V_0=V(a)$ & potential value at the global minimum \\
$V_1=V(-a)$ & potential value at the local minimum \\
$\beta=V(0)$ & barrier height in the quartic double well \\
$E$ & conserved classical ECD energy \\
$\lambda_c$ & Poisson flip rate in intrinsic $s$-time \\
$\lambda_q$ & qECD hitting time scaling parameter\\
$\sigma$ & qECD initialization Gaussian width\\
$\varepsilon$ & qECD hitting time detection window width\\
\bottomrule
\end{tabular}
\end{table}

\section{Proofs of Telegraph Propositions in Section \ref{subsec:telegraph}}
\label{sec:telegraph-proofs}
Let $\lambda = \lambda_c$.
We sketch the proofs of the boundary value problems with the same template: define the quantity of interest depending on state $(x, u)\in \mb{R}\times \{\pm\}$,
write $\mathcal L h=0$ or $\mathcal L G=-w$, impose boundary conditions,
and solve the resulting ODE system.
The generator on $(0,L)$ is
\begin{equation}
\begin{aligned}
(\mathcal L f)_+(x) &= f_+'(x)+\lambda(f_-(x)-f_+(x)),
\\
(\mathcal L f)_-(x) &= -f_-'(x)+\lambda(f_+(x)-f_-(x)).
\label{eq:telegraph-generator}
\end{aligned}
\end{equation}

\begin{proof}[Proof of \cref{prop:q}]
Let $h_\pm(x)=\mb P_{x,\pm}(X_T=L)$. On $(0,L)$, $\mathcal L h=0$ gives
\[
h_+'(x)+\lambda(h_-(x)-h_+(x))=0,\qquad
-h_-'(x)+\lambda(h_+(x)-h_-(x))=0.
\]
Note that
$h_-(0)=0$ and $h_+(L)=1$.
Subtract the equations to get $(h_+-h_-)'=0$, hence $h_+-h_-=D$ is constant.
Then $h_+'(x)=\lambda D$, so $h_+(x)=h_+(0)+\lambda D x$ and $h_-(x)=h_+(x)-D$.
From $h_-(0)=0$ we get $h_+(0)=D$, and from $h_+(L)=1$ we get $D(1+\lambda L)=1$,
so $q=h_+(0)=D=1/(1+\lambda L)$.
\end{proof}

\begin{proof}[Proof of \cref{lem:poisson-in}]
Recall that $G_-(0)=0$ and $G_+(L)=0$. Then, $\mathcal L G=-w$ on $(0,L)$ gives
\[
G_+'+\lambda(G_--G_+)=-w,\qquad -G_-'+\lambda(G_+-G_-)=-w,
\]
Set $D=G_+-G_-$ and $S=G_++G_-$, so
$D'=-2w$ and $S'=2\lambda D$. As $G_-(0)=0$, $c\coloneqq S(0)=D(0)$. We have
\[ c\coloneqq S(0)=D(0)\qquad S(L)=-D(L)\]
Let $W(x)=\int_0^x w(u)du$ and $V(x)=\int_0^x (x-u)w(u)du$, so 
\[
D(x)=c-2W(x)\implies S(x)=c+2\lambda \int_0^x D(u)du = c+2\lambda cx-4\lambda V(x)
\]
Now, as $G_+(L)=0$
\[ 0=D(L)+S(L) = c-2W(L)+c+2c\lambda L-4\lambda V(L)\implies c = \frac{W(L)+2\lambda V(L)}{1+\lambda L}\]
This is exactly $G_+(0)$. To solve for $G_-(L)$, we plug $c$ into $D(L)$.
\end{proof}
\begin{proof}[Proof of \cref{lem:poisson-out}]
Let $\sigma_L = \inf\{s>0:X_s\in\{0, L\}\}$, so by \cref{prop:q}
\[ \mb{P}_{(0, +)}(\sigma=\infty)=\lim_{L\to\infty}\mb{P}_{(0, +)}(X_{\sigma_L}=L)=\lim_{L\to\infty}\frac{1}{1+\lambda L}=0\]
Hence, $\sigma$ is finite almost surely. Conditioned on this event, the path maximum $M=\sup_{s\in [0, \sigma]}X_s$ is finite, so along each sample path $\int_0^{\sigma_L} w(X_s)ds =\int_0^\sigma w(X_s)ds$ for any $L>M$ is eventually constant. By monotone convergence as $L\to\infty$
\[\mb{E}_{(0, +)}\left[\int_0^\sigma w(X_s)ds\right]=\lim_{L\to\infty} \mb{E}_{(0, +)}\left[\int_0^{\sigma_L} w(X_s)ds\right]  \]
To compute the right-hand side, we use \cref{lem:poisson-in,prop:q} and integrability of $w$ to compute
\begin{equation}
q\int_0^L \left(1+2\lambda (L-x)\right)w(x)\, dx = \frac{(1+2\lambda L)\int_0^L w(x)dx-2\lambda \int_0^L xw(x)\, dx}{1+\lambda L}\to 2\int_0^\infty w(x)dx
\end{equation}
as $L\to\infty$, where the fact that $\frac{1}{L}\int_0^L xw(x)dx\to 0$ as $L\to\infty$ for integrable $w$ is a standard exercise in real analysis.
\end{proof}
\label{sec:telegraph-proofs}

\section{Proofs of Other Results in Section \ref{sec:classical}}

\begin{proof}[Proof of \cref{prop:Tdet}]
Since $u_s\equiv 1$, then $dx/ds=1$, so by \cref{eq:time-change,eq:T-S-conversion}
\begin{equation}
T_{\det} =\frac{1}{2}\int_0^S p(\phi^{-1}(X_s))^2ds
=\frac{1}{2}\int_0^L p(\phi^{-1}(x))^2dx
 = \frac{1}{2}\int_{-a}^a p(\theta)d\theta
\end{equation}
where we observe $dx/d\theta = \phi'(\theta)=1/p(\theta)$.
\end{proof}

\begin{proof}[Proof of \cref{prop:legs-hit}]
From $(0,+)$, with probability $q$ we hit the right well in one leg; with probability $1-q$ we return to $(0,-)$ in one leg. From $(0,-)$, we always take one outward-return leg to $(0,+)$. Thus, we have a linear system
\begin{equation}
\begin{aligned}
\mb{E}_{(0, +)}[D_{\mathrm{hit}}] & = q+(1-q)\left(1+\mb{E}_{(0, -)}[D_{\mathrm{hit}}]\right)\\
\mb{E}_{(0, -)}[D_{\mathrm{hit}}] & = 1+\mb{E}_{(0, +)}[D_{\mathrm{hit}}]
\end{aligned}
\end{equation}
which gives to the desired result. 
\end{proof}
\begin{proof}[Proof of \cref{prop:one-step-t}]
Recall that $w(x)=p(\phi^{-1}(x))^2/2$. By \cref{lem:poisson-out}, we compute 
\begin{equation}
\mb{E}_{(0, -)} [\Delta t]= 2\int_0^\infty w(x)dx =\int_{-\infty}^0 p(\phi^{-1}(x))^2dx= \int_{-\infty}^{-a}p(\theta)d\theta.
\end{equation}
where we observe $dx/d\theta = \phi'(\theta)=1/p(\theta)$. The same holds for the other outward state, i.e. $\mb{E}_{(L, +)} [\Delta t]$. 

From $(0, +)$, we recall $w$ to compute that
\begin{equation}
\mb{E}_{(0, +)} [\Delta t] = q\int_0^L (1+2\lambda (L-x))w(x)dx = \frac{q}{2}\int_0^L p(\phi^{-1}(x))^2 dx + 2q\lambda \int_0^L (L-x)w(x)dx = 
qT_{\det} + q\lambda B_{(0, +)}
\end{equation}
where we recognize the first term as $qT_{\det}$ by \cref{eq:Tdet} and the second term follows from
\begin{equation}\label{eq:B0+}
 2\int_0^L (L-x)w(x)dx = \int_0^L \left(\int_x^L dy\right)p(\phi^{-1}(x))^2dx
 = \int_{-a}^a \left(\int_\theta^a \frac{d\xi}{p(\xi)}\right)p(\theta)d\theta =B_{(0, +)}
\end{equation}
where $dx/d\theta = \phi'(\theta)=1/p(\theta)$ and similarly $dy/d\xi = 1/p(\xi)$. The same holds for $(L, -)$.
\end{proof}

\begin{proof}[Proof of \cref{prop:real-hit}]
We follow the proof in \cref{prop:legs-hit}.
The statement on stationary distribution follows from the fact that the embedded chain has uniform stationary distribution and the expected transition times given in \cref{prop:one-step-t}.
For the hitting times, starting at $(0, -)$, we always transition to $(0, +)$ in expected real time given exactly by $\mb{E}_{(0, -)}[\Delta t]$, so
\begin{equation}\label{eq:real-hit-eq-1}
\mb E_{(0,-)}[T_{\mathrm{hit}}] = \mb E_{(0,+)}[T_{\mathrm{hit}}] + \mb{E}_{(0, -)}[\Delta t].
\end{equation}
The first-step analysis starting at $(0, +)$ is more complicated: with probability $q$ we have the event $E$ that next state of $Z_n$ is $(L, +)$. Then, the expected hitting time is the expected time of this transition, i.e. $\mb{E}_{(0, +)}\left[\Delta t\middle| E \right]$. On the complement $\bar{E}$ of $E$, the next state is $(0, -)$, in which case the expected additional hitting time is $\mb E_{(0,-)}[T_{\mathrm{hit}}]$. Together, we have
\begin{equation}\label{eq:real-hit-eq-2}
\begin{aligned}
\mb E_{(0,+)}[T_{\mathrm{hit}}] & = q\mb E_{(0,+)}[\Delta t|E]+(1-q)\left(\mb E_{(0,+)}[\Delta t|\bar{E}]+\mb E_{(0,-)}[T_{\mathrm{hit}}]\right)
\\ & = (1-q)\mb E_{(0,-)}[T_{\mathrm{hit}}]+\mb{P}_{(0, +)}(E)\cdot\mb E_{(0,+)}[\Delta t|E]+\mb{P}_{(0, +)}(\bar{E})\cdot \mb E_{(0,+)}[\Delta t|\bar{E}]
\\ & = (1-q)\mb E_{(0,-)}[T_{\mathrm{hit}}]+\mb E_{(0,+)}[\Delta t].
\end{aligned}
\end{equation}
Solving \cref{eq:real-hit-eq-1,eq:real-hit-eq-2} gives the desired.
\end{proof}

\begin{proof}[Proof of \cref{thm:general-classical}]
Combining \cref{prop:real-hit,prop:one-step-t,prop:q}, we obtain
\begin{equation}
\begin{aligned}
\mb{E}_{(0,u_0)}[T_{\mathrm{hit}}] & =  
\frac{1}{q}\left(\mb{E}_{(0, +)} [\Delta t]+(1-q\mbf{1}_{\{u_0=1\}})\mb{E}_{(0, -)} [\Delta t]\right)
\\ & = \frac{1}{q}\mb{E}_{(0, +)} [\Delta t]+\left(\frac{1}{q}-\mbf{1}_{\{u_0=1\}}\right)\mb{E}_{(0, -)} [\Delta t]
\\ & = T_{\det}+\lambda B_{(0, +)}+\left(\mbf{1}_{\{u_0=-1\}}+\lambda L\right)\int_{-\infty}^{-a} p(\theta)d\theta
\end{aligned}
\end{equation}
which is exactly \cref{thm:general-classical} upon plugging in $B_{(0, +)}$.
\end{proof}
\begin{proof}[Proof of \cref{cor:symm-case}]
Recall that $w(x)=p(\phi^{-1}(x))^2/2$. Since $V$ is symmetric about $0$, so is $p$, and so $w$ is symmetric about $L/2$. Therefore, by \cref{eq:B0+}, in $x$-space
\begin{equation}
B_{(0, +)} = 2\int_0^L (L-x)w(x)dx=2\int_0^L (L-x)w(L-x)dx=2\int_0^L yw(y)dy = B_{(L, -)}
\end{equation}
Moreover, by \cref{eq:Tdet}
\begin{equation}
B_{(0, +)} +   B_{(L, -)} = 2L\int_0^Lw(x)dx=L\int_{0}^L p(\phi^{-1}(x))^2dx = 2LT_{\det}
\end{equation}
Therefore, $B_{(0, +)} = B_{(L, -)} = LT_{\det}$. Together, by symmetry and \cref{eq:Tdet}
\begin{equation}
\mb{E}_{(0, u_0)}[T_{\mathrm{hit}}] = (1+\lambda L)T_{\det}+(\lambda L+ \mbf{1}_{\{u_0=-1\}})\int_a^\infty p(\theta)d\theta = 
(1+\lambda L)\int_{0}^a p(\theta)d\theta+(\lambda L+ \mbf{1}_{\{u_0=-1\}})\int_a^\infty p(\theta)d\theta
\end{equation}
which simplifies to the desired equation.
\end{proof}

\section{Proof of Results in Section \ref{subsec:qenergy}}
{\paragraph{Remarks on quantization uniqueness.} Quantization of the ECD Hamiltonian is not canonical because $\hat{\Theta}$ and $\hat{p}$ do not commute. This operator ordering ambiguity is a well-known issue \citep{PhysRevB.27.7547}. In this paper, the symmetric ordering we chose
\begin{equation}
    H_0 = \hat{p}V(\hat{\Theta})\hat{p} = -\hbar^2\partial_{\Theta}(V(\Theta)\partial_{\Theta})
\end{equation}
is not unique. More generally, one can consider a family of symmetric quantum Hamiltonians 
\begin{equation}
    H(a,b,c) = \frac{1}{2}(V^a\hat{p}V^b\hat{p}V^c+V^c\hat{p}V^b\hat{p}V^a), \quad a+b+c = 1,
\end{equation}
where $V = V(\hat{\Theta})>0$ and $\hat{p}=-i\hbar\partial_{\Theta}$. The choice $(a,b,c)=(0,1,0)$ gives $H_0$. Equivalently, the different orderings can be expressed in terms of $H_0$ with subleading $\mathcal{O}(\hbar^2)$ quantum potentials 
\begin{equation}
    H(a,b,c) = H_0 - \hbar^2\left[\frac{a+c}{2}V''(\Theta)-ac\frac{(V'(\Theta))^2}{V(\Theta)}\right].
\end{equation}}

\begin{proof}[Proof of \cref{lem:discrete}]
Setting $\lambda = \frac{E}{\hbar^2}$ as the rescaled energy, the time-independent Schr\"odinger equation \cref{eq:TISE} becomes
\begin{equation}\label{eq:sturmliouville}
    -(V\psi')'=\lambda\psi.
\end{equation}
Define the Liouville transformation under coordinate change
\begin{equation}
    y(\Theta)\coloneqq \int_0^{\Theta}\frac{d\theta}{\sqrt{V(\theta)}}.
\end{equation}
From \cref{subsec:assumptions}, $V(\theta)$ has super-quadratic tails as $|\theta|\rightarrow \infty$, so the original domain $\Theta \in \mathbb{R}$ gets compactified with this redefinition onto a finite interval $(y_-,y_+)$.
The wavefunction unitarily transforms via
\begin{equation}
    u(y)\coloneqq V(\Theta(y))^{1/4}\psi(\Theta(y)),  \;\; d\Theta = \sqrt{V}dy
\end{equation}
such that
\begin{equation}
    V\int_{\mathbb{R}}|\psi(\Theta)|^2d\Theta = \int_{y_-}^{y_+}|u(y)|^2 dy.
\end{equation}
\cref{eq:sturmliouville} then transforms into the Liouville normal form, a standard Schr\"odinger equation on a bounded interval $(y_-,y_+)$
\begin{equation}
    -u''(y)+Q(y)u(y)=\lambda u(y),
\end{equation}
where $Q(y)$ is an effective potential
\begin{equation}\label{eq:Qy}
    Q(y)=\frac{1}{4}V''(\Theta(y))-\frac{1}{16}\frac{(V'(\Theta(y)))^2}{V(\Theta(y))} = \frac{1}{4}\frac{V_{yy}}{V}-\frac{3}{16}\frac{(V_y)^2}{V^2}.
\end{equation}
Therefore the resolvent is compact by standard spectral theorem result, so the spectrum of $H$ is \textit{pure point and discrete}.

To show semi-boundedness of the spectrum, for any $\psi \in C^{\infty}_c(\mathbb{R})$, since $V(x) > 0$ we have
\begin{equation}
    \langle \psi, H\psi\rangle=\int_{\mathbb{R}}\overline{\psi(x)}(-\hbar^2\partial_x(V(x)\psi'(x)))dx \stackrel{\text{ibp}}{=} \hbar^2 \int_{\mathbb{R}} V(x)|\psi'(x)|^2dx \geq 0.
\end{equation}
Using the Friedrichs extension associated with the closed quadratic form $q(\psi)=\langle \psi, H\psi\rangle \geq 0$, we have a Hamiltonian that is self-adjoint and bounded below by 0. Therefore $\sigma(H) \subset [0,\infty)$.
\end{proof}

\begin{proof}[Proof of \cref{lem:wkb}]
    For $E \geq E_{cut}$, the WKB eigenstates take the form of a plane wave ansatz
    \begin{equation}
        \psi(\Theta) = \exp[\frac{i}{\hbar}\phi(\Theta)],
    \end{equation}
    where $\phi(\Theta)$ is a complex phase. Expanding $\phi$ in power series in $\hbar$ up to $ O (\hbar)$, such that $\phi = \phi_0+\hbar\phi_1+O(\hbar^2)$, and substituting into \cref{eq:TISE} gives
    \begin{equation}\label{truncatedSE}
    i\hbar\frac{d}{d\Theta}\left[V(\Theta)(\phi_0'(\Theta)+\hbar\phi_1'(\Theta)+ O (\hbar^2))\right]-V(x)\left[\phi_0'(\Theta)+\hbar\phi_1'(\Theta)+ O (\hbar^2)\right]^2+E=0
    \end{equation}
    Solving the leading order $ O (\hbar^0)$ terms of the Schr\"odinger equation assuming $\hbar |\phi_1'(\Theta)| \ll |\phi_0'(\Theta)|$ (we state the exact condition at \cref{eq:wkb1}) gives the phase contribution to the eigenstate in terms of an integral of the local momentum $p(\Theta) := \sqrt{E/V(\Theta)}$
    \begin{equation}
    \phi_0(\Theta)=\pm\int^\Theta p(x)dx.
    \end{equation}
    Taking $\phi_0' = \pm p(\Theta)$ such that the leading order term vanishes, solving the next order $ O (\hbar)$ terms gives the amplitude of the eigenstate
    \begin{equation}
        \phi_1(\Theta) = \frac{i}{4}\ln E V(\Theta)+ const.
    \end{equation}
    Combining both terms and assuming $\hbar^2 |\phi_2'(\Theta)| \ll |\phi_0'(\Theta)|$ so the higher-order phase contribution is asymptotically small (we state the exact condition at \cref{eq:wkb2}), we obtain the desired equation. For completeness, we also solve the next order $ O (\hbar^2)$ terms and present it here\footnote{In the WKB eigenstate, this gives the next order contribution to the phase term, scaling as $\exp{ O (\hbar)}$.}
    \begin{equation}
        \phi_2(\Theta) = \mp \int^{\Theta}\frac{1}{2\sqrt{EV(x)}}\left[\frac{1}{4}V''(x)-\frac{1}{16}\frac{(V'(x))^2}{V(x)}\right].
    \end{equation}
    Interestingly, the numerator of the integrand is exactly the effective quantum potential in the transformed y-space. This confirms the intuition that the next order phase correction to the semiclassical approximation is a purely quantum one.

    Due to \cref{lem:discrete}, the energy eigenstates of the qECD Hamiltonian are bound states given by a linear combination of the orthogonal traveling-wave solutions 
    \begin{equation}
        \psi_{WKB,E}(\Theta) = \frac{A}{(EV(\Theta))^{1/4}}\sin\left(\frac{\sqrt{E}}{\hbar}\int_{\Theta_0}^{\Theta}\frac{1}{\sqrt{V(x)}}dx +  O (\hbar)+ \phi\right).
    \end{equation}
    where $A, \phi$ are amplitude and phase constants to be determined by normalization and boundary conditions, respectively. 
    The hard wall boundary conditions as enforced by the Friedrichs extension require the wavefunction to vanish at infinity
    \begin{equation}
        \lim_{\Theta \rightarrow \pm \infty} \psi_{WKB,E}(\Theta) = 0.
    \end{equation}
    This means that the sine functions must be integer multiples of $\pi$ at the boundaries, leading to the quantization rule where $n = 1,2,3, \cdots$
    \begin{equation}
        \frac{\sqrt{E_n}}{\hbar}L = n\pi, \qquad E_n = \left(\frac{n\pi\hbar}{L}\right)^2.
    \end{equation}
    This confirms the Liouville coordinate picture, where the system is effectively a particle in a rigid box of length $L$.
    
    $A$ is a normalization constant such that the wavefunction satisfies
    \begin{equation}
        \int_{-\infty}^{\infty} |\psi_n(x)|^2 dx = 1.
    \end{equation}
    Therefore the wavefunction simplifies to the desired expression
    \begin{equation}
        \psi_{WKB,n}(\Theta) = \sqrt{\frac{2}{L}}\frac{1}{V(\Theta)^{1/4}}\sin\left(\frac{n\pi}{L}\int_{-\infty}^{\Theta}\frac{1}{\sqrt{V(x)}}dx +  O (\hbar)\right).
    \end{equation}
To the $ O (\hbar)$ order in the phase expansion, the semiclassical approximation holds when the perturbation series for the phase converges rapidly. This requires the following conditions to be satisfied:
\begin{enumerate}
    \item In solving the leading order terms in the Schr\"odinger equation, we assumed that $\hbar |\phi_1'(\Theta)| \ll |\phi_0'(\Theta)|$. This reduces to
    \begin{equation}\label{eq:wkb1}
        \frac{\hbar \abs{V'(\Theta)}}{4\sqrt{E V(\Theta)}}\ll 1.
    \end{equation}
    Physically, this validity condition means that $V(\Theta)$ changes slowly over the distance of one local de Broglie wavelength.
    \item In dropping higher order corrections to the phase term, we assumed that $\hbar^2 |\phi_2'(\Theta)| \ll |\phi_0'(\Theta)|$. This reduces to
    \begin{equation}\label{eq:wkb2}
        \hbar^2\abs{\frac{V''(\Theta)}{4}-\frac{V'(\Theta)^2}{16V(\Theta)}} \ll 2E.
    \end{equation}
    Physically, this validity condition means that the second-order phase correction is small compared to the leading-order classical phase contribution. Interestingly, in the Liouville coordinate, this condition implies that the effective quantum Hamiltonian $Q(y)$ is negligible compared to the kinetic term. 
\end{enumerate}
The WKB condition for qECD can therefore be summarized as
\begin{equation}
    \max\left(\frac{\hbar\abs{V'(\Theta)}}{4\sqrt{EV(\Theta)}},\frac{\hbar^2\abs{V''(\Theta)}}{4E}\right) \ll 1. 
\end{equation}
In this work, we are interested in studying the transition between local minima under the qECD dynamics. To separate low energy from semiclassical physics, we define the semiclassical (WKB) energy cutoff within the region of interest
\begin{equation}
    E_{cut} \coloneqq \hbar^2 \max\left\{\sup_{x \in [-2a, 2a]}\frac{\abs{V'(x)}^2}{V(x)},\sup_{x \in [-2a, 2a]}\abs{V''(x)}\right\}.
\end{equation}
\end{proof}

\section{Proof of Results in Section \ref{subsec:qprop}}\label{sec:app-pf-4}
\begin{proof}[Proof of \cref{lem:Klow}]
     We start with the discrete spectrum expression for $K_{high}$, assuming no degeneracy:
\begin{equation}
    K_{low}(\Theta_2,t;\Theta_1,0) = \sum_{n: E_{n} < E_{cut}} e^{-iE_nt/\hbar}\psi_n(\Theta_2)\psi_n^*(\Theta_1).
\end{equation}
By Weyl's law, we can estimate the number of energy eigenstates below $E_{cut}$, $N_{low}$, via the WKB quantization rule up to an $ O (1)$ error, giving 
\begin{equation}
    N_{low}  = \frac{L}{2\pi} \max\left\{\sup_{x \in [-2a, 2a]}\frac{\abs{V'(x)}^2}{V(x)},\sup_{x \in [-2a, 2a]}\abs{V''(x)}\right\} +  O (1).
\end{equation}
Observe that the number of states residing in the non-WKB regime is independent of $\hbar$ at leading order. Assuming the geometric features of the landscape are $ O (1)$ and $\psi_n$ have unit norm, we have
\begin{equation}
    K_{low}(\Theta_2,t;\Theta_1,0) \leq N_{low} \sup \abs{e^{-iE_nt/\hbar}\psi_n(\Theta_2)\psi_n^*(\Theta_1)} =  O (1).
\end{equation}
\end{proof}
\begin{proof}[Proof of \cref{lem:Kwkb}]
WLOG, assume $\Theta_2 > \Theta_1$. We start with the discrete spectrum expression for $K_{wkb}$, assuming no degeneracy:
\begin{equation}
    K_{wkb}(\Theta_2,t;\Theta_1,0) = \sum_{n: E_{n}\geq E_{cut}} e^{-iE_nt/\hbar}\psi_n(\Theta_2)\psi_n^*(\Theta_1).
\end{equation}
Substituting the WKB wavefunctions derived in \cref{lem:wkb}, we have
\begin{equation}
    K_{wkb}(\Theta_2,t;\Theta_1,0) = \frac{2}{L(V(\Theta_2)V(\Theta_1))^{1/4}}\sum_{n: E_{n}\geq E_{cut}} e^{-iE_nt/\hbar}\sin\left(\sqrt{E_n}I(\Theta_2)+  O (\hbar)\right)\sin\left(\sqrt{E_n}I(\Theta_1) +  O (\hbar)\right),
\end{equation}
where we define the phase constant
\begin{equation}
    I(\Theta) := \int_{-\infty}^{\Theta}\frac{1}{\sqrt{V(x)}}dx.
\end{equation}
We can expand the product of the sine functions
\begin{equation}
    \sin\left(\theta_n(\Theta_2)\right)\sin\left(\theta_n(\Theta_1) \right) = \frac{1}{4}\left[ \exp\frac{iS_0}{\hbar} + \exp\frac{-iS_0}{\hbar} - \exp\frac{iS_l}{\hbar} - \exp\frac{-iS_l}{\hbar}\right],
\end{equation}
where $S_0$, $S_L$ are action-like quantities defined via
\begin{equation}
    S_0: = \hbar\sqrt{E_n}(I(\Theta_2)-I(\Theta_1)), \qquad S_l: = \hbar\sqrt{E_n} (I(\Theta_2)+I(\Theta_1)).
\end{equation}
Physically, $S_0$ denotes the action of a direct path from $\Theta_1$ to $\Theta_2$, while $S_l$ denotes the action of a path starting from $\Theta_1$ going into negative infinity and reflecting back to $\Theta_2$ ("left-reflecting"). To extract analytical behavior of this infinite sum, one can utilize the Poisson summation formula
\begin{equation}
    \sum_{n: E_n>E_{cut}}^{\infty} f(n) = \sum_{k = -\infty}^{\infty}\int_{E_{cut}}^{\infty}dE \left(\frac{dn}{dE}\right) f(n)e^{2\pi i k n}
\end{equation}
to rewrite the sum over discrete energy level $n$ as an integral over energy $E$. The density of state $dn/dE$ can be derived from the quantization condition 
\begin{equation}
    S(E) = \sqrt{E}L = n\pi\hbar \quad \Rightarrow dn = \frac{L}{2\sqrt{E}\pi\hbar}dE.
\end{equation}
Substituting back to the $K_{wkb}$ expression gives
\begin{equation}
    K_{wkb}(\Theta_2,t;\Theta_1,0) = \sum_{k=-\infty}^{\infty}\int_{E_{cut}}^{\infty}dE\frac{e^{-iE_nt/\hbar}e^{i2kS(E)/\hbar}}{4\pi\hbar\sqrt
    E(V(\Theta_2)V(\Theta_1))^{1/4}}\left[ \exp\frac{iS_0}{\hbar} + \exp\frac{-iS_0}{\hbar} - \exp\frac{iS_l}{\hbar} - \exp\frac{-iS_l}{\hbar}\right].
\end{equation}
The factor $e^{-i2kS(E)/\hbar}$ introduced by the Poisson summation formula physically represents the path action accumulated during $k$ round trips. The sum over $k$ can therefore be interpreted as summing the propagating waves over all winding numbers. For every $k$, the four terms represent a full set of topological reflections as follows:
\begin{enumerate}
    \item $S_0+2kS$ ("direct" path): $\Theta_1 \rightarrow \Theta_2$ + $k$ full round trips
    \item $-S_0+2kS$ ("left,right-reflected" path): $\Theta_1 \rightarrow -\infty \rightarrow +\infty \rightarrow \Theta_2$ + $(k-1)$ full round trips
    \item $S_l + 2kS$ ("left-reflected" path): $\Theta_1 \rightarrow -\infty \rightarrow \Theta_2$ + $k$ full round trips
    \item $-S_l+2kS$ ("right-reflected" path) : $\Theta_1 \rightarrow +\infty \rightarrow \Theta_2$ + $(k-1)$ full round trips.
\end{enumerate}
Therefore we see that the Poisson summation expression indeed sums over all possible paths that the quantum system can take between $\Theta_1$ and $\Theta_2$. 

It is easy to see that every summand in this infinite sum takes the general form
\begin{equation}
    J_{path} = \int_{E_{cut}}^{\infty}dE\frac{1}{4\pi\hbar\sqrt
    E(V(\Theta_2)V(\Theta_1))^{1/4}}\exp{\frac{i}{\hbar}\Phi(E)}, \qquad \Phi(E) = S_{path}-Et.
\end{equation}
The integral can be analyzed via a saddle point expansion
\begin{equation}
    \Phi'(E_*) = 0\quad \Rightarrow \quad E_* = \frac{I_{path}^2}{4t^2}, \quad I_{path} = \frac{S_{path}}{\sqrt{E}} = \int_{path}ds\frac{1}{\sqrt{V(s)}}.
\end{equation}
It is crucial for the saddle point approximation that $E_* >  E_{cut}$. Recalling the definition of $E_{cut}$ due to the WKB condition, we have
\begin{equation}
    \inf_{path}\frac{I^2_{path}}{4t^2} > \hbar^2 \max\left\{\sup_{x \in [-2a, 2a]}\frac{\abs{V'(x)}^2}{V(x)},\sup_{x \in [-2a, 2a]}\abs{V''(x)}\right\}
\end{equation}
Therefore, it suffices to ensure $t \le \frac{\delta I_0}{\hbar}$
for some $\delta>0$ depending only on landscape $V$ (and not $\hbar$). 
\textit{This sets the time scale where this semiclassical analysis remains valid, as stated in the lemma.}

Assuming that the time scale condition is satisfied, at the saddle point,
\begin{equation}
    \Phi(E_*) = \frac{I_{path}^2}{4 t}, \qquad \Phi''(E_*) = -\frac{2t^3}{ I_{path}^2}<0.
\end{equation}
Expanding the integral around the saddle point with $E = E_*+\epsilon$, we can keep terms up to $O(\epsilon^2)$ in the integrand and perform Gaussian integration to get
\begin{equation}
    J_{path}=\frac{1}{4\sqrt{\pi\hbar t}(V(\Theta_2)V(\Theta_1))^{1/4}}\exp{i\left(\frac{I_{path}^2}{4\hbar t}-\frac{\pi}{4}\right)}(1+ O (\hbar)).
\end{equation}
Note that the amplitude prefactor is independent of $k$. To further simplify the summands, observe that
\begin{equation}
    \sum_{k=-\infty}^{\infty} \exp{\frac{i}{4\hbar t}I^2_{direct}(k)} = \sum_{k=-\infty}^{\infty} \exp{\frac{i}{4\hbar t}(S_0+2kS)^2} = \sum_{k=-\infty}^{\infty} \exp{\frac{i}{4\hbar t}I^2_{left,right}(-k)}
\end{equation}
\begin{equation}
    \sum_{k=-\infty}^{\infty} \exp{\frac{i}{4\hbar t}I^2_{left}(k)} = \sum_{k=-\infty}^{\infty} \exp{\frac{i}{4\hbar t}(S_l+2kS)^2} = \sum_{k=-\infty}^{\infty} \exp{\frac{i}{4\hbar t}I^2_{right}(-k)}
\end{equation}
Collecting all terms:
\begin{equation}
\begin{split}
        &K_{wkb}(\Theta_2,t;\Theta_1,0)\\&= \frac{e^{-i\pi/4}}{2\sqrt{\pi\hbar t}(V(\Theta_2)V(\Theta_1))^{1/4}}\sum_{k=-\infty}^{\infty}\left[\exp{\frac{i(I(\Theta_2)-I(\Theta_1)+2kL)^2}{4\hbar t}}-\exp{\frac{i(I(\Theta_2)+I(\Theta_1)+2kL)^2}{4\hbar t}}\right]\\
        &=\frac{e^{-i\pi/4}}{2\sqrt{\pi\hbar t}(V(\Theta_2)V(\Theta_1))^{1/4}}\left[\exp{\frac{i(I(\Theta_2)-I(\Theta_1))^2}{4\hbar t}}\vartheta_3(z_-,\tau)-\exp{\frac{i(I(\Theta_2)+I(\Theta_1))^2}{4\hbar t}}\vartheta_3(z_+,\tau)\right],
\end{split}
\end{equation}
where $\vartheta_3(z_\pm,\tau)$ are Jacobi Theta functions with
\begin{equation}
    z\pm = \frac{L(I(\Theta_2)\pm I(\Theta_1))}{2\hbar t}, \qquad \tau = \frac{L^2}{\pi\hbar t}.
\end{equation}
We now show that only the $k=0$ term contributes to the evolution of a physical wave function, to leading order in $\hbar$. Now, recall from \cref{eq:evolvedstate} that an initial wave function evolves under the Hamiltonian via a spatial integral over the propagator kernel
\begin{equation}
    \psi(\Theta_2,t) \coloneqq \int d\Theta_1 K(\Theta_2,t;\Theta_1,0)\psi_0(\Theta_1).
\end{equation}
For the set up considered in this work, $\psi_0$ has zero momentum, so the integral contains fast oscillating terms taking the form
\begin{equation}
    \int d\Theta_1 (const)\cdot \frac{\psi(\Theta_1)}{V(\Theta_1)^{1/4}}\exp{\frac{i(I(\Theta_2) \pm I(\Theta_1)+2kL)^2}{4\hbar t}},
\end{equation}
where the constant term contains amplitude information of $\psi_0$ and $K_{wkb}$. The leading order contribution to this oscillatory integral can be analyzed using a saddle point expansion. The saddle points for these terms, for any fixed $k$, take the form
\begin{equation}
    I(\Theta^*_{2,k}) = \pm I(\Theta_1) \pm 2kL.
\end{equation}
Since $I(\Theta) \in [0,L]$, as discussed in the Liouville transformation section, the only physically permitted saddle point is $I(\Theta^*_{2,0}) = I(\Theta_1)$. This saddle point comes from the $k=0$ direct path term. Every other saddle point lies outside the integration domain, and therefore has negligible contribution to the dynamics of the quantum state. Therefore the effective time-propagator is given by
\begin{equation}
    K_{wkb}(\Theta_2,t;\Theta_1,0)= \frac{e^{-i\pi/4}}{2\sqrt{\pi\hbar t}(V(\Theta_2)V(\Theta_1))^{1/4}}\exp{\frac{i(I(\Theta_2)-I(\Theta_1))^2}{4\hbar t}+ O (\hbar)},
\end{equation}
thus proving the desired time-propagator result.
\end{proof}
\remark[The Continuum Approximation] Remarkably, the requirement that the first saddle point sits above the WKB energy cutoff sets a time scale where the semiclassical physics exhibits identical behavior to a system with a continuous spectrum, as can be seen by examining the $k=0$, direct path term in the Poisson sum before performing the saddle point expansion. 

\begin{proof}[Proof of \cref{lem:psit}]
    Starting from \cref{eq:evolvedstate}, we substitute the initial Gaussian state defined in \cref{eq:psi0} and the time-propagator derived in \cref{cor:propagatorapprox} to get the time-evolved quantum state. Let $\sigma = \alpha\sqrt{\hbar}$, where $\alpha$ is an $ O (1)$ initialization constant we can choose. The state at time $t$ is
    \begin{equation}
        \psi(\Theta,t) = \frac{e^{-i\pi/4}}{2\sqrt{\pi\hbar t}(2\pi\alpha^2\hbar V(\Theta))^{1/4}}\int_{-\infty}^{\infty}d\Theta' \frac{1}{(V(\Theta'))^{1/4}}\exp{\frac{1}{\hbar}\Phi(\Theta')}, \qquad \Phi(\Theta') = -\frac{(\Theta+a)^2}{4\alpha^2}+i\frac{I(\Theta',\Theta)^2}{4t}.
    \end{equation}
    The integrand is a sharp Gaussian envelope with an oscillatory part. The derivatives of the phase are 
    \begin{equation}
        \Phi'(\Theta') = -\frac{\Theta'+a}{2\alpha^2} - \frac{iI(\Theta',\Theta)}{2t\sqrt{V(\Theta)}}, \qquad \Phi''(\Theta) = -\frac{1}{2\alpha^2}+\frac{i}{2tV(\Theta')}+\frac{iI(\Theta',\Theta)V'(\Theta')}{4tV(\Theta')^{3/2}}.
    \end{equation}
    The saddle point of this integral can be expressed as an implicit equation
    \begin{equation}
        \Theta_*' = -a -\frac{i\alpha^2}{t\sqrt{V(\Theta'_*)}}I(-a,\Theta').
    \end{equation}
    The phase function and its second derivative at the saddle point are then
    \begin{equation}
        \Phi(\Theta_*') = \frac{\alpha^2I(-a,\Theta)^2}{4t^2V(\Theta'_*)}+\frac{iI(-a,\Theta)^2}{4t}.
    \end{equation}
    \begin{equation}
        \Phi''(\Theta_*') = -\frac{1}{2\alpha^2}+\frac{i}{2tV(\Theta'_*)}+\frac{\alpha^2I(-a,\Theta)^2V'(\Theta'_*)}{4t^2V(\Theta'_*)}.
    \end{equation}
    As $\hbar \rightarrow 0$, this integral is overwhelmingly dominated by an $ O (\hbar)$ region around the saddle point. Since $V(\Theta)$ is independent of $\hbar$, within this effective integration region, 
    \begin{equation}
        V(\Theta')^{-1/4} = V(\Theta_x')^{-1/4}(1+ O (\hbar)),
    \end{equation}
    so we can evaluate the slowly varying amplitude $V(\Theta')^{-1/4}$ at the saddle point and pull it in front of the integral with $ O (\hbar)$ error. The state at time $t$ can therefore be simplified into
    \begin{equation}
        \psi(\Theta,t) = \frac{e^{-i\pi/4}\exp{\frac{1}{\hbar}\Phi(\Theta'_*)}}{2\sqrt{\pi\hbar t}(2\pi\alpha^2\hbar V(\Theta)V_1)^{1/4}}\int_Cd\Theta'\exp{\frac{1}{2\hbar}\Phi''(\Theta_*')(\Theta'-\Theta'_*)^2+\frac{1}{24\hbar}\Phi^{(4)}(\Theta_*')(\Theta'-\Theta'_*)^4+\cdots},
    \end{equation}
    where we integrate over the contour of steepest descent $C$ across the complex saddle point $\Theta_*'$. Let $x = \Theta'-\Theta'_*$. The contour integral is dominated by the leading order Gaussian integral, with subleading contribution coming from the quartic term:
    \begin{equation}
        I_C = \int_Cdx \exp{\frac{1}{2\hbar}\Phi''(\Theta_*')x^2+\frac{1}{24\hbar}\Phi^{(4)}(\Theta_*')x^4+\cdots}=\sqrt{\frac{2\pi\hbar}{-\Phi''(\Theta_*')}}(1+ O (\hbar)).
    \end{equation}
    Combining all ingredients, we arrive at the final expression of the quantum state at time $t$:
    \begin{equation}
        \psi(\Theta,t) = \frac{e^{-i\pi/4}}{2\sqrt{\pi\hbar t}(2\pi\alpha^2\hbar V(\Theta) V(\Theta'_*)^{1/4}}\sqrt{\frac{2\pi\hbar}{-\Phi''(\Theta_*')}}\exp{\frac{i}{\hbar}\left(\frac{\alpha^2I(-a,\Theta)^2}{4t^2V(\Theta'_*)}+\frac{iI(-a,\Theta)^2}{4t}\right)}\left[1+ O (\hbar)\right].
    \end{equation}
    The probability density of the quantum state is given by
    \begin{equation}
        \abs{\psi(\Theta,t)}^2 = \frac{1}{2t\sqrt{2\pi\hbar V({\Theta}) \abs{V(\Theta'_*)}}\abs{\Phi''(\Theta'_*)}}\exp{\frac{2}{\hbar}\Re{\Phi(\Theta'_*)}}\left[1+ O (\hbar)\right],
    \end{equation}
    giving the desired result.
\end{proof}

\begin{proof}[Proof of \cref{lem:qhit} and \cref{thm:general-quantum}]
    Recall that the instantaneous transition probability is given by
    \begin{equation}
        p_{\sigma}(t) = \int_{a-\sigma}^{a+\sigma}d\Theta\abs{\psi(\Theta,t)}^2.
    \end{equation}
    Since we are integrating around an interval $2\alpha\sqrt{\hbar}$ that only contributes to higher-order $\hbar$ corrections, the leading-order instantaneous probability density is effectively constant over the integration window. Therefore
    \begin{equation}
        p_{\sigma}(t) = (2\alpha\sqrt{\hbar})\Theta\abs{\psi(a,t)}^2 = \frac{1}{t\sqrt{2\pi V_0 \abs{V(\Theta'_*)}}\abs{\Phi''(\Theta'_*)}}\exp{\frac{2}{\hbar}\Re{\Phi(\Theta'_*)}}\left[1+ O (\hbar)\right].
    \end{equation}
    Recall that $\Im{\Theta'_*} = -\frac{\alpha^2 I(\Theta'_*,\Theta)}{t\sqrt{V(\Theta_*')}}$, so at time $O (1/\sqrt{\hbar})$ and above, the saddle point becomes asymptotically close to the real axis. We will later confirm that this is indeed the characteristic time scale of the dynamics. At this time scale, we can approximate
    \begin{equation}
        \abs{V(\Theta'_*)} = V_1 +  O (\hbar), \quad \abs{\Phi''(\Theta'_*)} = \frac{1}{2\alpha^2}+ O (\hbar), \quad \Re{\Phi(\Theta'_*)} = \frac{\alpha^2I(-a,a)^2}{2t^2V_1}+ O (\hbar).
    \end{equation}
    The average transition probability is therefore
    \begin{equation}
        \overline{p_{\sigma}}(\tau) = \int_0^{\tau}\frac{dt}{\tau} \frac{\sqrt{2}\alpha^2}{t\sqrt{2\pi V_0 V_1}}\exp{-\frac{\alpha^2 I(-a,a)^2}{2\hbar t^2 V_1}} dt  = \frac{A}{2\tau}\text{E}_1\left(\frac{B}{\tau^2}\right)\left[1+ O (\hbar)\right],
    \end{equation}
    where E$_1(x) := \int_x^{\infty}\frac{e^{-u}}{u}du$ is the exponential integral, and $A, B$ are
    \begin{equation}
        A = \frac{\alpha^2\sqrt{2}}{\sqrt{\pi V_1V_0}}, \qquad B = \frac{\alpha^2I(-a,a)^2}{\hbar V_1}.
    \end{equation}
    With a change of variable $x = B/\tau^2$,
    \begin{equation}
        \frac{\tau}{\overline{p_{\sigma}}(\tau)} = \frac{2B}{A}\frac{1}{x\text{E}_1(x)}\left[1+ O (\hbar)\right] = (const) \cdot \frac{I(-a,a)^2}{\hbar}\sqrt{\frac{V_0}{V_1}}\left[1+ O (\hbar)\right].
    \end{equation}
    Note that the minimum of $\frac{1}{x\text{E}_1(x)}$ can be reached at $x_0$, where $x_0$ is an $\mathcal{O}(\hbar)$ constant, so we can extract the leading-order characteristic time scale up to a constant 
    \begin{equation}
        \tau_* = (const)\cdot \frac{I(-a,a)}{\sqrt{\hbar V_1}}.
    \end{equation}
    This $\mathcal{O}(1/\sqrt{\hbar})$ time scale confirms the vanishing complex shift of the saddle point, and satisfies the limit imposed by \cref{eq:time-scale}. This proves the desired result in \cref{lem:qhit}. 

    The desired upper bound in \cref{thm:general-quantum} is achieved with $\tau_*$, with $\overline{p_{\sigma}}(\tau_*) = (const) \cdot \frac{\sqrt{\hbar}}{I(-a,a)\sqrt{V_0}}$ to leading order in $\hbar$. This concludes the quantum hitting time proof.
\end{proof}

\section{Case Study Computations in Section \ref{sec:comparison}}
We prove \cref{thm:comparison-small,thm:comparison-large} together.
By a change of variables, with $\delta = V_0/\beta$ and $\theta=au$, we have
\begin{equation}
V(au)=\beta\left(\delta+(1-u^2)^2\right)
\end{equation}
Therefore, the cases of guessing error correspond to $\delta\lesssim 1$ versus $\delta\gtrsim 1$. We compute integrals
\begin{equation}\label{eq:L-app}
L=\frac{a\sqrt{\beta}}{\sqrt{E}}\int_{-1}^1 \sqrt{\delta+(1-u^2)^2}du
\end{equation}
and as $a/\sqrt{\beta}\asymp 1/\omega$, we have
\begin{equation}\label{eq:p-app}
\int_a^\infty p(\theta)d\theta=\frac{a\sqrt{E}}{\sqrt{\beta}}\int_1^\infty \frac{du}{\sqrt{\delta+(1-u^2)^2}}\qquad \int_0^a p(\theta)d\theta=\frac{a\sqrt{E}}{\sqrt{\beta}}\int_0^1 \frac{du}{\sqrt{\delta+(1-u^2)^2}}
\end{equation}
Also, note that for $I$ defined in \cref{eq:I}
\[I(-a, a)\asymp E^{-1/2}\int_0^ap(\theta)d\theta\]
Hence, to compute $T_c$ and $T_q$, it suffices to compute the three integrals over $u$ above and apply \cref{cor:symm-case} and \cref{thm:general-quantum}. 
\begin{proof}[Proof of \cref{thm:comparison-small}]
Here, $\delta\lesssim 1$. Then, the integral in \cref{eq:L-app} is a constant. For the integrals in \cref{eq:p-app}, we note a divergence near $u=1$: let $t=u-1$, then $(u^2-1)^2=(2t+t^2)^2\sim 4t^2$ for $u\in [0, 2]$, so
\[
\int_0^1 \frac{du}{\sqrt{\delta+(1-u^2)^2}} \asymp \int_{1}^2 \frac{du}{\sqrt{\delta+(1-u^2)^2}} \asymp \int_{0}^C\frac{dt}{\sqrt{\delta+4t^2}}\asymp \log (\delta^{-1})
\]
This shows the second integral in \cref{eq:p-app}. For the first one, we split at $u=2$ and use the equation above as well as
\[
\int_2^\infty \frac{du}{\sqrt{\delta+(1-u^2)^2}} \le \int_2^\infty \frac{3du}{u^2}=O(1)
\]
since $\sqrt{\delta+(u^2-1)^2}\ge u^2-1\gtrsim u^2/3$ for $u\ge 2$. Together
\[
\int_1^\infty \frac{du}{\sqrt{\delta+(1-u^2)^2}}=\int_1^2 \frac{du}{\sqrt{\delta+(1-u^2)^2}}+\int_2^\infty \frac{du}{\sqrt{\delta+(1-u^2)^2}}\asymp \log (\delta^{-1})
\]
so both integrals in \cref{eq:p-app} must be asymptotically $\log(\delta^{-1})$.
\end{proof}
\begin{proof}[Proof of \cref{thm:comparison-large}]
 If $\delta\gtrsim 1$, then $(1-u^2)^2\in [0, 1]$ for $u\in [-1, 1]$, so
\[ \int_{-1}^1 \sqrt{\delta+(1-u^2)^2}du\asymp {\sqrt{\delta}}\quad\text{and}\quad\int_0^1 \frac{du}{\sqrt{\delta+(1-u^2)^2}}\asymp \frac{1}{\sqrt{\delta}}\]
Now, we split the last integral to bound the denominators and conclude 
\[
\int_1^\infty\frac{du}{\sqrt{\delta+(1-u^2)^2}} = \int_0^{\delta^{1/4}}\frac{du}{\sqrt{\delta+(1-u^2)^2}}+\int_{\delta^{1/4}}^\infty\frac{du}{\sqrt{\delta+(1-u^2)^2}}\le \int_0^{\delta^{1/4}}\frac{du}{\sqrt{\delta}} + \int_{\delta^{1/4}}^\infty\frac{du}{u^2-1}\lesssim \delta^{-1/4}
\]
and we obtain an asymptotically matching lower bound:
\[
\int_1^\infty\frac{du}{\sqrt{\delta+(1-u^2)^2}} \ge \int_0^{\delta^{1/4}}\frac{du}{\sqrt{\delta+(1-u^2)^2}}\gtrsim \int_0^{\delta^{1/4}}\frac{du}{\sqrt{\delta+u^4}}\gtrsim\int_0^{\delta^{1/4}}\frac{du}{\sqrt{\delta}}=\delta^{-1/4}
\]
so the tail integral in $u$ must be asymptotically $\delta^{-1/4}$.
\end{proof}

\section{qECD Simulation Details}
\label{sec:numerics}

{The analytical results for hitting times of sECD and qECD are verified numerically in \cref{fixed V0} and \cref{scaled V0}. The code can be found in \href{https://github.com/jean-wang-huaijin/Classical-and-Quantum-ECD}{this repository}. Here, we discuss implementation details of the qECD simulation.

The qECD simulation discretizes the simulated spatial region $[-R,R]$ into $N$ grid points at positions $x_j$. The initial state is a normalized Gaussian wave packet with  default center $x_0 = -a$ defined on these grid points. Throughout this paper, we chose $R=6$ and $N = 25000$. The potential landscape $V(x)$ is evaluated on midpoint grid points $x_{j+1/2}$, creating a tridiagonal stencil
\begin{equation}
    H\psi_j = \frac{V_{j-1/2}+V_{j+1/2}}{\Delta x^2}\psi_j - \frac{V_{j+1/2}}{\Delta x^2}\psi_{j+1} - \frac{V_{j-1/2}}{\Delta x^2}\psi_{j-1}.
\end{equation}
The discretized time evolution is then performed using the Crank-Nicolson method to the interior grid points, with Dirichlet boundary conditions
\begin{equation}
    \left(I + \frac{i H \Delta t }{2}\right)\psi^{n+1} = \left(I - \frac{iH\Delta t}{2}\right)\psi^n.
\end{equation}
The computational domain is chosen to be sufficiently large such that boundary effects are negligible. The initial state is a normalized Gaussian wavepacket centered at the left minimum. The Gaussian width is chosen as $\sigma = 0.1/\omega$, so that the wavefunction stays localized relative to the narrowing well as $\omega$ increases. The detection window is chosen to be the same as the initial Gaussian width. \cref{fig:qECD heatmap} visualizes the time evolution of the wavefunction probability density under the qECD Hamiltonian for a chosen symmetric double-well potential.}

\begin{figure}[ht]
  \vskip 0.2in
  \begin{center}
    \centerline{\includegraphics[width=0.8\columnwidth]{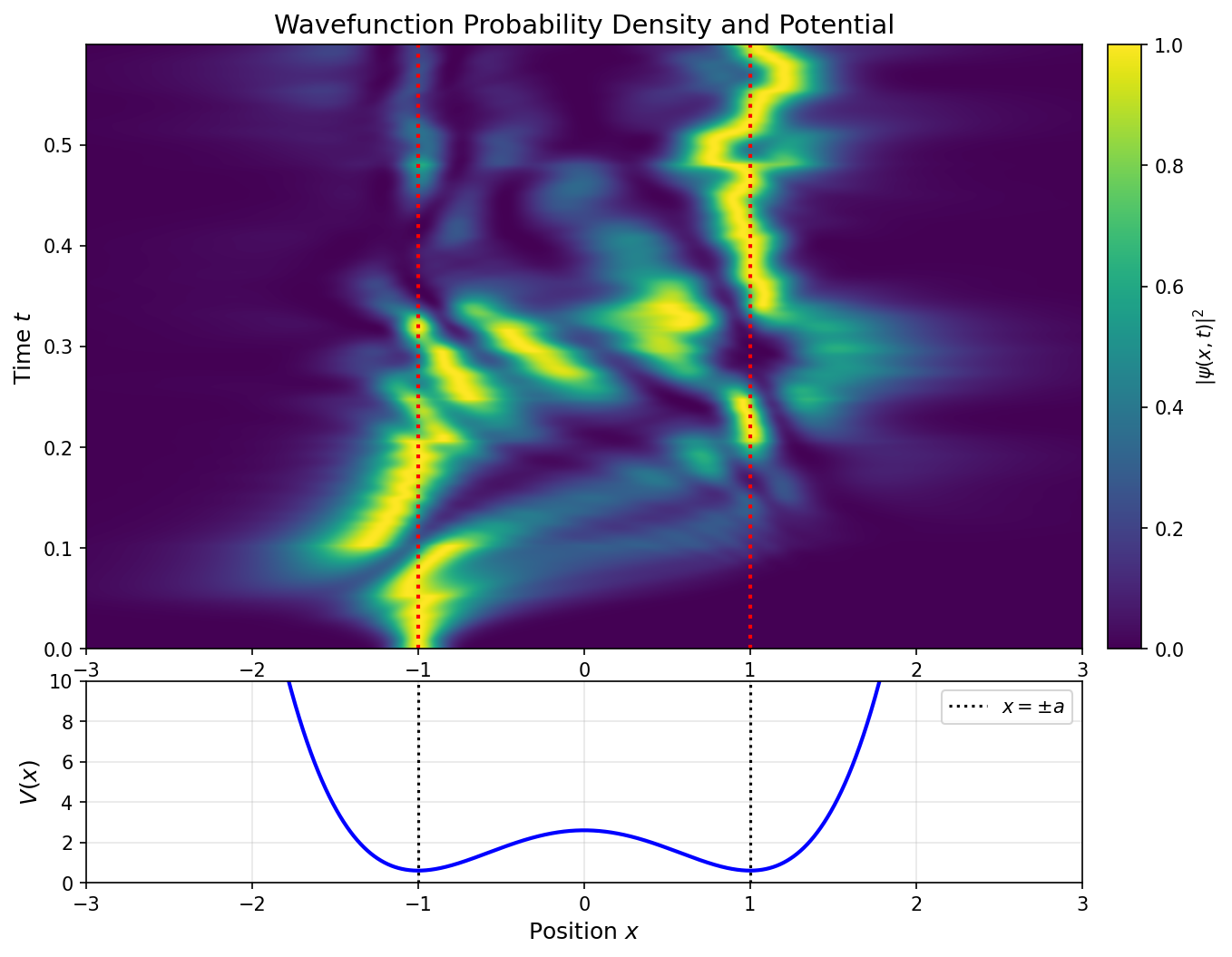}}
    \caption{Time evolution of the wavefunction probability density under the symmetrical double-well potential, simulated over $x \in [-6,6]$. The top panel shows the normalized probability density $|\psi(x,t)|^2$ for $x \in [-3,3]$; red dotted lines mark the location of the wells. The bottom panel shows the corresponding potential $V(x)=2(x^2-1)^2+0.6$. }
    \label{fig:qECD heatmap}
  \end{center}
\end{figure}

\end{document}